\documentclass{lmcs}
\pdfoutput=1

\usepackage{lastpage}
\lmcsdoi{15}{2}{17}
\lmcsheading{}{\pageref{LastPage}}{}{}%
{Jan.~31,~2017}{May~23,~2019}{}

\usepackage{microtype,hyperref}
\usepackage{xspace,amsmath,amsfonts,amssymb,bbm,ifthen,mathtools}
\usepackage[heavycircles]{stmaryrd}
\usepackage{bbold}
\usepackage{enumerate}
\usepackage[sort]{cite}
\usepackage{tikz}
\usetikzlibrary{arrows,positioning,automata,shapes,calc,patterns}
\usepackage{tikz-3dplot}
\tikzset{state/.style={shape=circle, draw, initial text=, inner
    sep=.5mm, minimum size=5mm}}
\tikzset{state with output/.style={shape=circle split, draw, initial
    text=, inner sep=.9mm, minimum size=15mm}}

\newcommand*\ttto[2]{\ensuremath{\mathrel{\smash{%
      \xrightarrow[%
      {\raisebox{.5ex}[0pt][0pt]{%
          \hspace*{.3em}$\scriptstyle#2$\hspace*{.3em}}}%
      ]{#1}}}}}

\newcommand*\Nat{\ensuremath{\mathbbm{N}}}

\newcommand*\Real{\ensuremath{\mathbbm{R}}}
\newcommand*\Realnn{\ensuremath{\Real_{\ge0}}}
\newcommand*\Realnp{\ensuremath{\Real_{\le0}}}
\newcommand*\Realni{\ensuremath{[ 0, \infty]}}

\newcommand*\bigmid{\mathrel{\big|}}
\newcommand*\ie{\textit{i.e.,}\xspace}

\newcommand*\ltrue{\textup{\textbf{t\!t}}}
\newcommand*\lfalse{\textup{\textbf{ff}}}

\newlength{\myautolength}
\newlength{\mysecondautolength}
\newlength{\mythirdautolength}
\newlength{\minautolength}
\newcommand*{\MaxLength}[3]{%
\ifthenelse{\lengthtest{\the#1>#2}}
           {\setlength{#3}{#1}}
           {\setlength{#3}{#2}}}
\newcommand*\acceptingornot{}

\newcommand{\simpleauto}[4][0]{
\ifthenelse{\equal{#1}{}}{\renewcommand*\acceptingornot{accepting}}{\renewcommand*\acceptingornot{}}
\settowidth{\myautolength}{$#4$}
\addtolength{\myautolength}{1em}
\MaxLength{\myautolength}{\minautolength}{\myautolength}
\begin{tikzpicture}[baseline, shorten >=1pt]
  \node[state, initial]                      (1)      {$#2$};
  \node[state, \acceptingornot, right=\myautolength of 1] (end)    {$#1$};
  \path[->] (1)  edge    node[above] {$#3$} node[below] {$#4$} (end);
\end{tikzpicture}
}

\newcommand{\doubleauto}[7][0]{
\ifthenelse{\equal{#1}{}}{\renewcommand*\acceptingornot{accepting}}{\renewcommand*\acceptingornot{}}
\settowidth{\myautolength}{$#4$}
\addtolength{\myautolength}{1em}
\MaxLength{\myautolength}{\minautolength}{\myautolength}
\settowidth{\mysecondautolength}{$#7$}
\addtolength{\mysecondautolength}{1em}
\MaxLength{\mysecondautolength}{\minautolength}{\mysecondautolength}
\begin{tikzpicture}[baseline, shorten >=1pt]
  \node[state, initial]            (1)                          {$#2$};
  \node[state, right=\myautolength of 1]            (2)           {$#5$};
  \node[state, \acceptingornot, right=\mysecondautolength of 2] (end) {$#1$};
  \path[->] (1)  edge    node[above] {$#3$} node[below] {$#4$} (2);
  \path[->] (2)  edge    node[above] {$#6$} node[below] {$#7$} (end);
\end{tikzpicture}
}

\newcommand{\tripleauto}[9]{
\settowidth{\myautolength}{$#3$}
\addtolength{\myautolength}{1em}
\MaxLength{\myautolength}{\minautolength}{\myautolength}
\settowidth{\mysecondautolength}{$#6$}
\addtolength{\mysecondautolength}{1em}
\MaxLength{\mysecondautolength}{\minautolength}{\mysecondautolength}
\settowidth{\mythirdautolength}{$#9$}
\addtolength{\mythirdautolength}{1em}
\MaxLength{\mythirdautolength}{\minautolength}{\mythirdautolength}
\begin{tikzpicture}[baseline, shorten >=1pt]
  \node[state, initial]            (1)                          {$#1$};
  \node[state, right=\myautolength of 1]            (2)           {$#4$};
  \node[state, right=\mysecondautolength of 2]            (3)           {$#7$};
  \node[state, accepting, right=\mythirdautolength of 3] (end) {};
  \path[->] (1)  edge    node[above] {$#2$} node[below] {$#3$} (2);
  \path[->] (2)  edge    node[above] {$#5$} node[below] {$#6$} (3);
  \path[->] (3)  edge    node[above] {$#8$} node[below] {$#9$} (end);
\end{tikzpicture}
}

\newcommand*\F{\mathcal{F}}
\newcommand*\G{\mathcal{G}}
\newcommand*\A{\mathcal{A}}
\newcommand*\E{\mathcal{E}}
\newcommand*\V{\mathcal{V}}

\newcommand*\compop{\mathrel{\mathop{\triangleright}}} 
\newcommand*\comp[3]{#1\compop#2(#3)}
\newcommand*\U{\mathcal{U}}
\newcommand*\one{\mathbf{1}}
\newcommand*\Bool{\mathbbm{B}}
\newcommand*\Ax[1]{$(\textrm{\textup{C}}#1)$}

\makeatletter
\pgfdeclarepatternformonly[\LineSpace]{my north east lines}{\pgfqpoint{-1pt}{-1pt}}{\pgfqpoint{\LineSpace}{\LineSpace}}{\pgfqpoint{\LineSpace}{\LineSpace}}%
{
    \pgfsetcolor{\tikz@pattern@color}
    \pgfsetlinewidth{0.4pt}
    \pgfpathmoveto{\pgfqpoint{0pt}{0pt}}
    \pgfpathlineto{\pgfqpoint{\LineSpace+0.1pt}{\LineSpace+0.1pt}}
    \pgfusepath{stroke}
}

\pgfdeclarepatternformonly[\LineSpace]{my north west lines}{\pgfqpoint{-1pt}{-1pt}}{\pgfqpoint{\LineSpace}{\LineSpace}}{\pgfqpoint{\LineSpace}{\LineSpace}}%
{
    \pgfsetcolor{\tikz@pattern@color}
    \pgfsetlinewidth{0.4pt}
    \pgfpathmoveto{\pgfqpoint{0pt}{\LineSpace}}
    \pgfpathlineto{\pgfqpoint{\LineSpace+0.1pt}{-0.1pt}}
    \pgfusepath{stroke}
}
\makeatother
\newdimen\LineSpace%
\tikzset{
    line space/.code={\LineSpace=#1},
    line space=3pt
}
\pgfdeclarepatternformonly{my dots}{\pgfqpoint{-1pt}{-1pt}}{\pgfqpoint{1pt}{1pt}}{\pgfqpoint{2pt}{2pt}}%
{
  \pgfpathcircle{\pgfqpoint{0pt}{0pt}}{.5pt}
  \pgfusepath{fill}
}

\begin{document}

\title{An \texorpdfstring{$\boldsymbol\omega$}{Omega}-Algebra for Real-Time Energy Problems}

\author[D.~Cachera]{David Cachera\rsuper{a}}%
\address{\lsuper{a}Universit{\'e} de Rennes, Inria, CNRS, IRISA, France}

\author[U.~Fahrenberg]{Uli Fahrenberg\rsuper{b}}%
\address{\lsuper{b}{\'E}cole Polytechnique, Palaiseau, France}%
\thanks{Most of this work was completed while the second author was still
  employed at Irisa / Inria Rennes.}

\author[A.~Legay]{Axel Legay\rsuper{c}}%
\address{\lsuper{c}Universit{\'e} catholique de Louvain, Belgium}

\titlecomment{This is a revised and extended version of the
  paper~\cite{DBLP:conf/fsttcs/CacheraFL15} which has been presented
  at the 35th IARCS Annual Conference on Foundation of Software
  Technology and Theoretical Computer Science (FSTTCS 2015) in
  Bangalore, India.  Compared to~\cite{DBLP:conf/fsttcs/CacheraFL15},
  and in addition to numerous small changes and improvements,
  motivation and examples as well as proofs of all results have been
  added to the paper.}

\begin{abstract}
  We develop a $^*$-continuous Kleene $\omega$-algebra of real-time
  energy functions.  Together with corresponding automata, these can
  be used to model systems which can consume and regain energy (or
  other types of resources) depending on available time.  Using recent
  results on $^*$-continuous Kleene $\omega$-algebras and
  computability of certain manipulations on real-time energy
  functions, it follows that reachability and B{\"u}chi acceptance in
  real-time energy automata can be decided in a static way which only
  involves manipulations of real-time energy functions.
\end{abstract}

\maketitle

\section{Introduction}

\emph{Energy} and \emph{resource management} problems are important in
areas such as embedded systems or autonomous systems.  They are
concerned with the following types of questions:
\begin{itemize}
\item \textit{Can the system reach a designated state without running
    out of energy before?}
\item \textit{Can the system reach a designated state within a
    specified time limit without running out of energy?}
\item \textit{Can the system repeatedly accomplish certain designated
    tasks without ever running out of energy?}
\end{itemize}
Instead of energy, these questions can also be asked using other
resources, for example money or fuel.

\begin{figure}
  \centering
  \includegraphics[width=.7\linewidth, trim=0 60 0 0,
  clip]{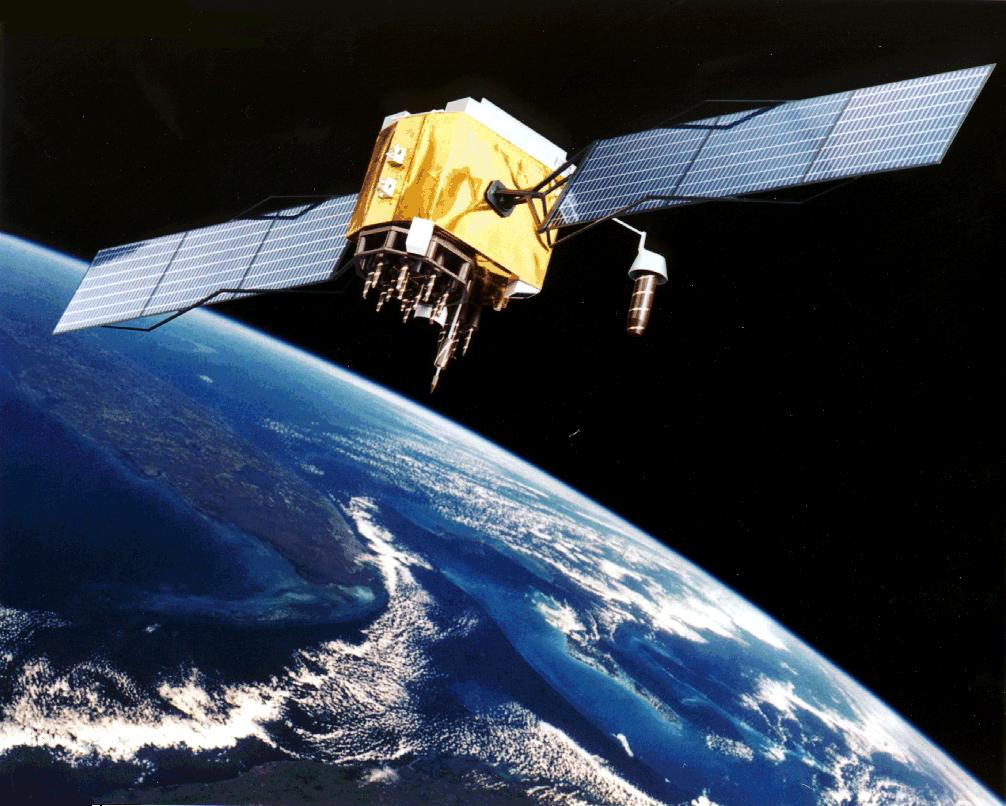}
  \caption{GPS Block II-F satellite (artist's conception; public domain)}\label{fi:sat}
\end{figure}

As an example, imagine a satellite like in Fig.~\ref{fi:sat} which is
being launched into space.  In its initial state when it has arrived at
its orbit, its solar panels are still folded, hence no (electrical)
energy is generated.  Now it needs to unfold its solar panels and
rotate itself and its panels into a position orthogonal to the sun's
rays (for maximum energy yield).  These operations require energy
which hence must be provided by a battery, and there may be some
operational requirements which state that they have to be completed
within a given time limit.  To minimize weight, one will generally be
interested to use a battery with minimal possible capacity.

\begin{figure}[bp]
  \centering
  \begin{tikzpicture}[shorten >=1pt, auto, ->, xscale=3, yscale=2.5]
    \node[state with output, initial] (00) at (0,0)
    {closed\vphantom{p} \nodepart{lower} $0$};
    \node[state with output] (10) at (1,0) {half\vphantom{p}
      \nodepart{lower} $2$};
    \node[state with output] (20) at (2,0) {open \nodepart{lower}
      $5$};
    \node[state with output] (01) at (0,-1) {closed\vphantom{p}
      \nodepart{lower} $0$};
    \node[state with output] (11) at (1,-1) {half\vphantom{p}
      \nodepart{lower} $4$};
    \node[state, accepting, align=center, minimum size=15mm] (21) at
    (2,-1) {opera- \\ tional};
    \path (00) edge node[below] {open} node[above] {$-20$} (10);
    \path (10) edge node[below] {open} node[above] {$-20$} (20);
    \path (01) edge node[below] {open} node[above] {$-20$} (11);
    \path (11) edge node[below] {open} node[above] {$-20$} (21);
    \path (00) edge node[left] {rotate} node[right] {$-10$} (01);
    \path (10) edge node[left] {rotate} node[right] {$-10$} (11);
    \path (20) edge node[left] {rotate} node[right] {$-10$} (21);
  \end{tikzpicture}
  \caption{Toy model of the satellite in Fig.~\ref{fi:sat}}\label{fi:sat-model}
\end{figure}
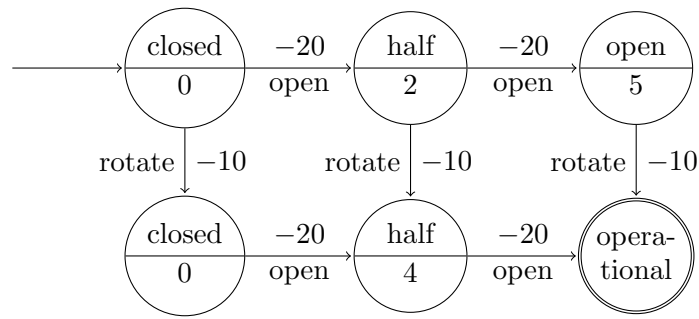

Figure~\ref{fi:sat-model} shows a simple toy model of such a
satellite's initial operations.  We assume that it opens its solar
panels in two steps; after the first step they are half open and
afterwards fully open, and that it can rotate into orthogonal position
at any time.  The numbers within the states signify energy gain per
time unit, so that for example in the half-open state, the satellite
gains $2$ energy units per time unit before rotation and $4$ after
rotation.  The (negative) numbers at transitions signify the energy
cost for taking that transition, hence it costs $20$ energy units to
open the solar panels and $10$ to rotate.

Now if the satellite battery has sufficient energy, then we can follow
any path from the initial to the final state without spending time in
intermediate states.  A simple inspection reveals that a battery level
of $50$ energy units is required for this.  On the other hand, if
battery level is strictly below $20$, then no path is available to the
final state.  With initial energy level between these two values, the
device has to regain energy by spending time in an intermediate state
before proceeding to the next one.  The optimal path then depends on
the available time and the initial energy.  For an initial energy
level of at least $40$, the fastest strategy consists in first opening
the panels and then spending $2$ time units in state (open|$5$) to
regain enough energy to reach the final state.  With the smallest
possible battery, storing $20$ energy units, $5$ time units have to be
spent in state (half|$2$) before passing to (half|$4$) and
spending another $5$ time units there.

In this paper we will be concerned with models for such systems which,
as in the example, allow to spend time in states to regain energy, of
which some has to be spent when taking transitions between states.
(Instead of energy, other resource types could be modeled, but we will
from now think of it as energy.)  We call these models \emph{real-time
  energy automata}.  Their behavior depends, thus, on both the initial
energy and the time available; as we have seen in the example, this
interplay between time and energy means that even simple models can
have rather complicated behaviors.  As in the example, we will be
concerned with the \emph{reachability} problem for such models, but
also with \emph{B{\"u}chi acceptance}: whether there exists an
infinite run which visits certain designated states infinitely often.

Our methodology is strictly algebraic, using the theory of
semiring-weighted auto\-mata~\cite{book/DrosteKV09} and extensions
developed in~\cite{conf/dlt/EsikFL15, DBLP:journals/corr/EsikFL15a}.
We view the finite behavior of a real-time energy automaton as a
function $f(x_0,t)$ which maps initial energy $x_0$ and available time
$t$ to a final energy level, intuitively corresponding to the highest
output energy the system can achieve when run with these parameters.
We define a composition operator on such \emph{real-time energy
  functions} which corresponds to concatenation of real-time energy
automata and show that with this composition and maximum as operators,
the set of real-time energy functions forms a \emph{$^*$-continuous
  Kleene algebra}~\cite{DBLP:journals/iandc/Kozen94}.  This implies
that reachability in real-time energy automata can be decided in a
static way which only involves manipulations of real-time energy
functions.

To be able to decide B{\"u}chi acceptance, we extend the algebraic
setting to also encompass real-time energy functions which model
infinite behavior.  These take as input an initial energy $x_0$ and
time $t$, as before, but now the output is Boolean: true if these
parameters permit an infinite run, false if they do not.  We show that
both types of real-time energy functions can be organized into a
\emph{$^*$-continuous Kleene $\omega$-algebra} as defined
in~\cite{conf/dlt/EsikFL15, DBLP:journals/corr/EsikFL15a}.  This
entails that also B{\"u}chi acceptance for real-time energy automata
can be decided in a static way which only involves manipulations of
real-time energy functions.

The most technically demanding part of the paper is to show that
real-time energy functions form a \emph{locally closed
  semiring}~\cite{book/DrosteKV09, journals/mfm/EsikK02}; generalizing
some arguments in~\cite{journals/mfm/EsikK02,
  DBLP:journals/corr/EsikFL15a}, it then follows that they form a
$^*$-continuous Kleene $\omega$-algebra.  We conjecture that
reachability and B{\"u}chi acceptance in real-time energy automata can
be decided in exponential time.

\subsection*{Related work}

Real-time energy problems have been considered
in~\cite{DBLP:conf/lata/Quaas11, DBLP:conf/formats/BouyerFLMS08,
  DBLP:conf/hybrid/BouyerFLM10, DBLP:journals/pe/BouyerLM14,
  DBLP:conf/ictac/FahrenbergJLS11}.  These are generally defined on
\emph{priced timed automata}~\cite{DBLP:conf/hybrid/AlurTP01,
  DBLP:conf/hybrid/BehrmannFHLPRV01}, a formalism which is more
expressive than ours: it allows for time to be reset and admits
several independent time variables (or \emph{clocks}) which can be
constrained at transitions.  All known decidability results apply to
priced timed automata with only \emph{one} clock;
in~\cite{DBLP:journals/pe/BouyerLM14} it is shown that with four
clocks, it is undecidable whether there exists an infinite run.

The work which is closest to ours
is~\cite{DBLP:conf/hybrid/BouyerFLM10}.  Their models are priced timed
automata with one clock and energy updates on transitions, hence a
generalization of ours.  Using a sequence of complicated ad-hoc
reductions, they show that reachability and existence of infinite runs
is decidable for their models; whether their techniques apply to
general B{\"u}chi acceptance is unclear.

Our work is part of a program to make methods from semiring-weighted
automata available for energy problems.  Starting
with~\cite{DBLP:conf/atva/EsikFLQ13}, we have developed a general
theory of \mbox{$^*$-continuous} Kleene
$\omega$-algebras~\cite{conf/dlt/EsikFL15,
  DBLP:journals/corr/EsikFL15a, DBLP:journals/actaC/EsikFLQ17,
  DBLP:journals/actaC/EsikFLQ17a} and shown that it applies to
so-called \emph{energy automata}, which are finite (untimed) automata
which allow for rather general \emph{energy transformations} as
transition updates.  The contribution of this paper is to show that
these algebraic techniques can be applied to a real-time setting.

Note that the application of Kleene algebra to real-time and hybrid
systems is not a new subject, see for
example~\cite{DBLP:journals/jlp/HofnerM09,
  DBLP:conf/RelMiCS/DongolHMS12}.  However, the work in these papers
is based on \emph{trajectories} and \emph{interval predicates},
respectively, whereas our work is on real-time energy \emph{automata},
\ie~at a different level.  A more thorough comparison of our work
to~\cite{DBLP:journals/jlp/HofnerM09, DBLP:conf/RelMiCS/DongolHMS12}
would be interesting future work.

\subsection*{Acknowledgment}

We are deeply indebted to our colleague and friend Zolt{\'a}n {\'E}sik
who taught us all we know about Kleene algebras and $^*$-continuity.
This work was started during a visit of Zolt{\'a}n at Irisa in Rennes;
unfortunately, Zolt{\'a}n did not live to see it completed.

\section{Real-Time Energy Automata}%
\label{se:rtea}

Let $\Realnn=[ 0, \infty\mathclose[$ denote the set of non-negative
real numbers, $\Realni$ the set $\Realnn$ extended with infinity, and
$\Realnp=\mathopen] -\infty, 0]$ the set of non-positive real numbers.

\begin{defi}
  A \emph{real-time energy automaton} (RTEA) $( S, s_0, F, T, r)$
  consists of a finite set $S$ of \emph{states}, with \emph{initial
    state} $s_0\in S$, a subset $F\subseteq S$ of \emph{accepting}
  states, a finite set
  $T\subseteq S\times \Realnp\times \Realnn\times S$ of
  \emph{transitions}, and a mapping $r: S\to \Realnn$ assigning
  \emph{rates} to states.  A transition $( s, p, b, s')$ is written
  $s\ttto p b s'$, $p$ is called its \emph{price} and $b$ its
  \emph{bound}.  We assume $b\ge -p$ for all transitions
  $s\ttto p b s'$.
\end{defi}

An RTEA is \emph{computable} if all its rates, prices and bounds are
computable real numbers.

A \emph{configuration} of an RTEA $A=( S, s_0, F, T, r)$ is an element
$( s, x, t)\in C= S\times[ 0, \infty]\times[ 0, \infty]$.  Let
$\mathord\leadsto\subseteq C\times C$ be the relation given by
$( s, x, t)\leadsto( s', x', t')$ iff $t'\le t$ and there is a
transition $s\ttto p b s'$ such that $x+( t- t') r( s)\ge b$ and
$x'= x+( t- t') r( s)+ p$.  Hence $t- t'$ time units are spent in
state $s$ and afterwards a transition $s\ttto p b s'$ is taken.

A \emph{run} in $A$ is a path in the infinite directed graph
$( C, \mathord\leadsto)$, \ie~a finite or infinite sequence
$( s_1, x_1, t_1)\leadsto( s_2, x_2, t_2)\leadsto\dotsm$.  We are
ready to state the decision problems for RTEAs with which we will be
concerned.  Let $A=( S, s_0, F, T, r)$ be a computable RTEA and
$x_0, t, y\in \Realni$ computable numbers.

\begin{prob}[State reachability]%
  \label{prb:reach}
  Does there exist a finite run
  $( s_0, x_0, t)\leadsto\dotsm\leadsto( s, x, t')$ in $A$ with
  $s\in F$?
\end{prob}

\begin{prob}[Coverability]%
  \label{prb:cover}
  Does there exist a finite run
  $( s_0, x_0, t)\leadsto\dotsm\leadsto( s, x, t')$ in $A$ with
  $s\in F$ and $x\ge y$?
\end{prob}

\begin{prob}[B{\"u}chi acceptance]%
  \label{prb:buchi}
  Does there exist $s\in F$ and an infinite run
  $( s_0, x_0, t)\leadsto( s_1, x_1, t_1)\leadsto\dotsm$ in $A$ in
  which $s_n= s$ for infinitely many $n\ge 0$?
\end{prob}

Note that the coverability problem only asks for the final energy
level $x$ to be \emph{above} $y$; as we are interested in
\emph{maximizing} energy, this is natural.  Also, state reachability
can be reduced to coverability by setting $y= 0$.  As the B{\"u}chi
acceptance problem asks for infinite runs, there is no notion of
output energy for this problem.

Asking the B{\"u}chi acceptance question for a \emph{finite} available
time $t< \infty$ amounts to finding (accepting) \emph{Zeno runs} in
the given RTEA, \ie~runs which make infinitely many transitions in
finite time.  Hence one will usually be interested in B{\"u}chi
acceptance only for an infinite time horizon.

On the other hand, for $t= \infty$, a positive answer to the state
reachability problem~\ref{prb:reach} will consist of a finite run
$( s_0, x_0, \infty)\leadsto\dotsm\leadsto( s, x, \infty)$.  Now as
one can delay indefinitely in the state $s\in F$, this yields an
infinite \emph{timed run} in the RTEA\@.  Per our definition of
$\leadsto$ however, such an infinite run will \emph{not} be a positive
answer to the B{\"u}chi acceptance problem.

\section{Weighted Automata over \texorpdfstring{$^*$}{Star}-Continuous Kleene
  \texorpdfstring{$\omega$}{Omega}-Algebras}

We now turn our attention to the algebraic setting of $^*$-continuous
Kleene algebras and related structures and review some results on
$^*$-continuous Kleene algebras and $^*$-continuous Kleene
$\omega$-algebras which we will need in the sequel.

\subsection{\texorpdfstring{$^*$}{Star}-Continuous Kleene Algebras}

An \emph{idempotent semiring}~\cite{book/Golan99}
$S=( S, \vee, \cdot, \bot, 1)$ consists of an idempotent commutative
monoid $( S, \vee, \bot)$ and a monoid $( S, \cdot, 1)$ such that the
distributive and zero laws
\begin{equation*}
  x( y\vee z)= x y\vee x z \qquad\qquad
  ( y\vee z) x= y x\vee z x \qquad\qquad
  \bot x= \bot= x \bot
\end{equation*}
hold for all $x, y, z\in S$. It follows that the product operation
distributes over all finite suprema.  Each idempotent semiring $S$ is
partially ordered by the relation $x\le y$ iff $x\vee y = y$, and then
sum and product preserve the partial order and $\bot$ is the least
element.

A \emph{Kleene algebra}~\cite{DBLP:journals/iandc/Kozen94} is an
idempotent semiring $S=( S, \vee, \cdot, \bot, 1)$ equipped with an
operation $^*: S\to S$ such that for all $x, y\in S$, $y x^*$ is the
least solution of the fixed point equation $z= z x\vee y$ and $x^* y$
is the least solution of the fixed point equation $z= x z\vee y$ with
respect to the order $\le$.

A \emph{$^*$-continuous Kleene
  algebra}~\cite{DBLP:journals/iandc/Kozen94} is a Kleene algebra
$S=( S, \vee, \cdot, ^*, \bot, 1)$ in which the infinite suprema
$\bigvee_{ n\ge 0} x^n$ exist for all $x\in S$,
$x^*= \bigvee_{ n\ge 0} x^n$ for every $x\in S$, and product preserves
such suprema: for all $x, y\in S$,
\begin{equation*}
  y \Big( \bigvee_{ n\ge 0} x^n \Big)= \bigvee_{ n\ge 0} y x^n\quad
  \text{and}\quad \Big( \bigvee_{ n\ge 0} x^n \Big) y= \bigvee_{ n\ge
    0} x^n y\,.
\end{equation*}

Examples of $^*$-continuous Kleene algebras include the set
$P( \Sigma^*)$ of languages over an alphabet $\Sigma$, with set union
as $\vee$ and concatenation as $\cdot$, and the set $P( A\times A)$ of
relations over a set $A$, with set union as $\vee$ and relation
composition as $\cdot$.  These are, in fact, \emph{continuous} Kleene
algebras in the sense that suprema $\bigvee X$ of arbitrary subsets
$X$ exist.

An important example of a $^*$-continuous Kleene algebra which is not
continuous is the set $R( \Sigma^*)$ of \emph{regular} languages over
an alphabet $\Sigma$.  This example is canonical in the sense that $R(
\Sigma^*)$ is the \emph{free} $^*$-continuous Kleene algebra over
$\Sigma$.

An idempotent semiring $S = (S,\vee, \cdot, \bot,1)$ is said to be
\emph{locally closed}~\cite{journals/mfm/EsikK02} if it holds that for
every $x\in S$, there exists $N\ge 0$ so that
$\bigvee_{ n= 0}^N x^n= \bigvee_{ n= 0}^{ N+ 1} x^n$.  In any locally
closed idempotent semiring, we may define a $^*$-operation by
$x^*= \bigvee_{ n\ge 0} x^n$.

\begin{lem}%
  \label{le:starcont}
  Any locally closed idempotent semiring is a $^*$-continuous Kleene
  algebra.
\end{lem}

\begin{proof}
  Let $S=( S, \mathord\vee, \mathord\cdot, \bot, 1)$ be a locally
  closed idempotent semiring.  We need to show that for all elements
  $x, y, z\in S$,
  \begin{equation*}
    x y^*= \bigvee_{ n\ge 0}( x y^n)
    \qquad\text{and}\qquad y^* z= \bigvee_{ n\ge 0}( y^n z)\,.
  \end{equation*}
  It is clear that the right-hand sides of the equations are less than
  or equal to their left-hand sides, so we are left with proving the
  other inequalities.  As $S$ is locally closed, there is $N\ge 0$
  such that $y^*= \bigvee_{ n= 0}^N y^n$, and then by distributivity,
  \[
    x y^*= x\Big( \bigvee_{ n= 0}^N y^n \Big)= \bigvee_{ n= 0}^N( x y^n)\le
  \bigvee_{ n\ge 0}( x y^n)\ ;
  \]
  similarly, $y^* z\le \bigvee_{ n\ge 0}( y^n z)$. 
\end{proof}

\subsection{\texorpdfstring{$^*$}{Star}-Continuous Kleene \texorpdfstring{$\omega$}{Omega}-Algebras}%
\label{se:starcontkleom}

An \emph{idempotent semiring-semimodule
  pair}~\cite{journals/sgf/EsikK07, book/BloomE93} $( S, V)$ consists
of an idempotent semiring $S=( S, \vee, \cdot, \bot, 1)$ and a
commutative idempotent monoid $V=( V, \vee, \bot)$ which is equipped
with a left $S$-action $S\times V\to V$, $( s, v)\mapsto s v$,
satisfying
\begin{alignat*}{3}
  ( s\vee s') v &= s v\vee s' v \qquad\qquad&
  s( v\vee v') &= s v\vee s v' \qquad\qquad&
  \bot v &= \bot \\
  (s s') v &= s( s' v) &
  s \bot &= \bot &
  1 v &= v
\end{alignat*}
for all $s, s'\in S$ and $v\in V$.  In that case, we also call $V$ a
\emph{(left) $S$-semimodule}.

A \emph{generalized $^*$-continuous Kleene
  algebra}~\cite{conf/dlt/EsikFL15} is an idempotent
semiring-semi\-module pair $( S, V)$ where $S=(S,\vee,\cdot, ^*,\bot,1)$
is a $^*$-continuous Kleene algebra such that for all $x, y\in S$ and
for all $v\in V$,
\begin{equation*}
  x y^* v= \bigvee_{ n\ge 0} x y^n v
\end{equation*}

A \emph{$^*$-continuous Kleene
  $\omega$-algebra}~\cite{conf/dlt/EsikFL15} consists of a generalized
$^*$-continuous Kleene algebra $( S, V)$ together with an infinite
product operation $S^\omega\to V$ which maps every infinite sequence
$x_0, x_1,\dotsc$ in $S$ to an element $\prod_{ n\ge 0} x_n$ of $V$.
The infinite product is subject to the following conditions:
\begin{itemize}
\item For all $x_0, x_1,\dotsc\in S$,
  \begin{equation}
  \prod_{ n\ge 0} x_n= x_0 \prod_{ n\ge 0} x_{ n+
    1},  \tag*{\Ax1}
  \end{equation}
\item Let $x_0, x_1,\dotsc\in S$ and $0= n_0\le n_1\le\dotsm$ a
  sequence which increases without a bound. Let
  $y_k= x_{ n_k}\dotsm x_{ n_{ k+ 1}- 1}$ for all $k\ge 0$.  Then
  \begin{equation}
  \prod_{ n\ge 0} x_n= \prod_{ k\ge 0} y_k \tag*{\Ax2}
  \end{equation}
\item For all $x_0, x_1,\dotsc, y, z\in S$,
  \begin{equation}
  \prod_{ n\ge 0}( x_n( y\vee z))= \adjustlimits \bigvee_{ x_0', x_1',\dotsc\in\{ y, z\}\;} \prod_{ n\ge 0} x_n x_n' \tag*{\Ax3}
  \end{equation}
\item For all $x, y_0, y_1,\dotsc\in S$,
  \begin{equation}
  \prod_{ n\ge 0} x^* y_n= \adjustlimits \bigvee_{ k_0, k_1,\dotsc\ge 0\;} \prod_{ n\ge 0} x^{ k_n} y_n \tag*{\Ax4}
  \end{equation}
\end{itemize}
Hence the infinite product extends the finite product~\Ax1; it is
finitely associative~\Ax2; it preserves finite suprema~\Ax3; and it
preserves the $^*$-operation (and hence infinite suprema of the form
$\bigvee_{ n\ge 0} x^n$)~\Ax4.  The infinite product gives rise to an
\emph{$\omega$-operation} $^\omega: S\to V$ defined by $x^\omega=
\prod_{ n\ge 0} x$.

An example of a $^*$-continuous Kleene $\omega$-algebra is the
structure $( P( \Sigma^*), P( \Sigma^\infty))$ consisting of the set
$P( \Sigma^*)$ of languages of finite words and of the set $P(
\Sigma^\infty)$ of finite or infinite words over an alphabet
$\Sigma$.  This is, in fact, a \emph{continuous} Kleene
$\omega$-algebra~\cite{conf/dlt/EsikFL15} in the sense that the
infinite product preserves \emph{all} suprema.

A $^*$-continuous Kleene $\omega$-algebra which is \emph{not}
continuous is $( R( \Sigma^*), R'( \Sigma^\infty))$, where
$R( \Sigma^*)$ is the set of regular languages over $\Sigma$, and
$R'( \Sigma^\infty)$ contains all subsets of the set $\Sigma^\infty$
of finite or infinite words which are finite unions of \emph{finitary}
infinite products of regular languages.  This is in fact the
\emph{free} finitary $^*$-continuous Kleene $\omega$-algebra over
$\Sigma$, see~\cite{conf/dlt/EsikFL15}.

\subsection{Matrix Semiring-Semimodule Pairs}

For any idempotent semiring $S$ and $n\ge 1$, we can form the matrix
semiring $S^{ n\times n}$ whose elements are $n\times n$-matrices of
elements of $S$ and whose sum and product are given as the usual
matrix sum and product.  It is known~\cite{DBLP:conf/mfcs/Kozen90}
that when $S$ is a $^*$-continuous Kleene algebra, then
$S^{ n\times n}$ is also a $^*$-continuous Kleene algebra, with the
$^*$-operation defined by
\begin{equation}
  \label{eq:mstar-sup}
  M^*_{ i, j}= \bigvee_{ m\ge 0} \bigvee\big\{ M_{ k_1, k_2} M_{ k_2,
    k_3}\dotsm M_{ k_{ m- 1}, k_m}\bigmid 1\le k_1,\dotsc, k_m\le n,
  k_1= i, k_m= j\big\}
\end{equation}
for all $M\in S^{ n\times n}$ and $1\le i, j\le n$.  Also, if $n\ge 2$
and $M= \left( \begin{smallmatrix} a & b \\ c & d \end{smallmatrix}
\right)$, where $a$ and $d$ are square matrices, then
\begin{equation}
  \label{eq:mstar}
  M^*=
  \begin{pmatrix}
    {( a\vee b d^* c)}^* & {( a\vee b d^* c)}^* b d^* \\
    {( d\vee c a^* b)}^* c a^* & {( d\vee c a^* b)}^*
  \end{pmatrix}.
\end{equation}

For any idempotent semiring-semimodule pair $( S, V)$ and $n\ge 1$, we
can form the matrix semiring-semimodule pair $( S^{ n\times n}, V^n)$
whose elements are $n\times n$-matrices of elements of $S$ and
$n$-dimensional (column) vectors of elements of $V$, with the action
of $S^{ n\times n}$ on $V^n$ given by the usual matrix-vector product.

When $( S, V)$ is a $^*$-continuous Kleene $\omega$-algebra, then
$( S^{ n\times n}, V^n)$ is a generalized $^*$-continuous Kleene
algebra~\cite{conf/dlt/EsikFL15}.
By~\cite[Lemma~17]{conf/dlt/EsikFL15}, there is an $\omega$-operation
on $S^{ n\times n}$ defined by
\begin{equation*}
  M^\omega_i= \bigvee_{1\le k_1,k_2,\dotsc\le n} M_{ i, k_1} M_{ k_1, k_2}\dotsm
\end{equation*}
for all $M\in S^{ n\times n}$ and $1\le i\le n$.  Also, if $n\ge 2$
and $M= \left( \begin{smallmatrix} a & b \\ c & d \end{smallmatrix}
\right)$, where $a$ and $d$ are square matrices, then
\begin{equation}
  \label{eq:momega}
  M^\omega=
  \begin{pmatrix}
    {( a\vee b d^* c)}^\omega\vee{( a\vee b d^* c)}^* b d^\omega \\
    {( d\vee c a^* b)}^\omega\vee{( d\vee c a^* b)}^* c a^\omega
  \end{pmatrix}.
\end{equation}

It can be shown~\cite{book/EsikK07} that the number of semiring
computations required in the computation of $M^*$ and $M^\omega$
in~\eqref{eq:mstar} and~\eqref{eq:momega} is $O( n^3)$ and $O( n^4)$,
respectively.

\subsection{Weighted automata}%
\label{se:weightedaut}

Let $( S, V)$ be a $^*$-continuous Kleene $\omega$-algebra and
$A\subseteq S$ a subset.  We write $\langle A\rangle$ for the set of
all finite suprema $a_1\vee\dotsm\vee a_m$ with $a_i\in A$ for each
$i= 1,\dotsc, m$.

A \emph{weighted automaton}~\cite{book/DrosteKV09} over $A$ of
dimension $n\ge 1$ is a tuple $( \alpha, M, k)$, where
$\alpha\in{\{ \bot, 1\}}^n$ is the initial vector,
$M\in \langle A \rangle^{ n\times n}$ is the transition matrix, and
$k$ is an integer $0\le k\le n$.  Combinatorially, this may be
represented as a transition system whose set of states is
$\{ 1,\dotsc, n\}$.  For any pair of states $i, j$, the transitions
from $i$ to $j$ are determined by the entry $M_{ i, j}$ of the
transition matrix: if $M_{ i, j}= a_1\vee\dotsm\vee a_m$, then there
are $m$ transitions from $i$ to $j$, respectively labeled
$a_1,\dotsc, a_m$.  The states $i$ with $\alpha_i= 1$ are
\emph{initial}, and the states $\{ 1,\dotsc, k\}$ are
\emph{accepting}.

The \emph{finite behavior} of a weighted automaton $A=( \alpha, M, k)$
is defined to be
\begin{equation*}
  | A|= \alpha M^* \kappa\,,
\end{equation*}
where $\kappa\in{\{ \bot, 1\}}^n$ is the vector given by $\kappa_i= 1$
for $i\le k$ and $\kappa_i= \bot$ for $i> k$.  (Note that $\alpha$ has
to be used as a \emph{row} vector for this multiplication to make
sense.)  It is clear by~\eqref{eq:mstar-sup} that $| A|$ is the supremum
of the products of the transition labels along all paths in $A$ from
any initial to any accepting state.

The \emph{B{\"u}chi behavior} of a weighted automaton
$A=( \alpha, M, k)$ is defined to be
\begin{equation*}
  \| A\|= \alpha
  \begin{pmatrix}
    {( a\vee b d^* c)}^\omega \\
    d^* c{( a\vee b d^* c)}^\omega
  \end{pmatrix},
\end{equation*}
where $a\in \langle A\rangle^{ k\times k}$,
$b\in \langle A\rangle^{ k\times( n- k)}$,
$c\in \langle A\rangle^{( n- k)\times n}$ and
$d\in \langle A\rangle^{( n- k)\times( n- k)}$ are such that
$M= \left( \begin{smallmatrix} a & b \\ c &
    d \end{smallmatrix}\right)$.
Note that $M$ is split in submatrices
$\left( \begin{smallmatrix} a & b \\ c & d \end{smallmatrix}\right)$
precisely so that $a$ contains transitions between accepting states
and $d$ contains transitions between non-accepting states.
By~\cite[Thm.~20]{conf/dlt/EsikFL15}, $\| A\|$ is the supremum of the
products of the transition labels along all infinite paths in $A$ from
any initial state which infinitely often visit an accepting state.

\section{Real-Time Energy Functions}

We are now ready to consider the algebra of real-time energy
functions.  We will build this up inductively, starting from the
functions which correspond to simple \emph{atomic} RTEAs.  These can
be composed to form \emph{linear} real-time energy functions, and
with additional maximum and star operations, they form a
$^*$-continuous Kleene algebra.  When also taking infinite behaviors
into account, we get a $^*$-continuous Kleene $\omega$-algebra of
real-time energy functions.

Let $L={[ 0, \infty]}_\bot$ denote the set of non-negative real numbers
extended with a bottom element $\bot$ and a top element $\infty$.  We
use the standard order on $L$, \ie~the one on $\Realnn$ extended by
declaring $\bot\le x\le \infty$ for all $x\in L$.  $L$ is a complete
lattice, whose suprema we will denote by $\vee$ for binary and
$\bigvee$ for general supremum.  For convenience we also extend the
addition on $\Realnn$ to $L$ by declaring that
$\bot+ x= x+ \bot= \bot$ for all $x\in L$ and
$\infty+ x= x+ \infty= \infty$ for all $x\in L\setminus\{ \bot\}$.
Note that $\bot+ \infty= \infty+ \bot= \bot$.

Let $\F$ denote the set of monotonic functions
$f: L\times \Realni\to L$ (with the product order on
$L\times \Realni$) for which $f( \bot, t)= \bot$ for all
$t\in[ 0, \infty]$.  We will frequently write such functions in
curried form, using the isomorphism
$\langle L\times \Realni\to L\rangle\approx\langle \Realni\to L\to
L\rangle$.

\subsection{Linear Real-Time Energy Functions}%
\label{se:lrtef}

We will be concerned with the subset of functions in $\F$ consisting
of \emph{real-time energy functions} (RTEFs).  These correspond to
functions expressed by RTEAs, and we will construct them inductively.
We start with atomic RTEFs:

\begin{defi}
  Let $r, b, p\in \Real$ with $r\ge 0$, $p\le 0$ and $b\ge -p$.
  An \emph{atomic real-time energy function} is an element $f$ of $\F$
  such that $f( \bot, t)= \bot$, $f( \infty, t)= \infty$,
  $f( x, \infty)= \infty$, and
  \begin{equation*}
    f( x, t)=
    \begin{cases}
      x+ r t+ p &\text{if } x+ r t\ge b\,, \\
      \bot &\text{otherwise}
    \end{cases}
  \end{equation*}
  for all $x, t\in \Realnn$.  The numbers $r, b$ and $p$ are
  respectively called the \emph{rate}, \emph{bound} and \emph{price}
  of $f$. We denote by $\A\subseteq \F$ the set of atomic real-time
  energy functions.
\end{defi}

These functions arise from RTEAs with one transition:
\begin{equation*}
  \simpleauto[]{r}{p}{b}
\end{equation*}

Non-negativity of $r$ ensures that atomic RTEFs are monotonic. In our
examples, when the bound is not explicitly mentioned it corresponds to
the lowest possible one: $b= -p$.

Atomic RTEFs are naturally combined along acyclic paths by means of a
composition operator. Intuitively, a composition of two successive
atomic RTEFs determines the optimal output energy one can get after
spending some time in either one or both locations of the
corresponding automaton. This notion of composition is naturally
extended to all functions in $\F$, and formally defined as follows
(where $\circ$ denotes standard function composition).

\begin{defi}
  The \emph{composition} of $f, g\in \F$ is the element $f\compop g$
  of $\F$ such that
  \begin{equation}
    \label{eq:comp}
    \forall t\in \Realni:( f\compop g)( t)= \bigvee_{ t_1+ t_2= t} g(
    t_2)\circ f( t_1)
  \end{equation}
\end{defi}

Note that composition is written in diagrammatic order.  Uncurrying
the equation, we see that $( f\compop g)( x, t)= \bigvee_{ t_1+ t_2= t}
g( f( x, t_1), t_2)$.

\begin{rem}%
  \label{re:notassoc}
  Composition in $\F$ is not generally associative:\footnote{The
    authors thank an anonymous reviewer for this example.}  Let
  $f, g, h\in \F$ be the functions given by
  \begin{gather*}
    f( x, t)=
    \begin{cases}
      \bot &\text{if } x= \bot\,,\\
      t &\text{otherwise}\,,
    \end{cases}
    \qquad g( x, t)=
    \begin{cases}
      \bot &\text{if } x= \bot\,,\\
      0 &\text{if } x= 0\,,\\
      2 t &\text{otherwise}\,,
    \end{cases}
    \qquad h( x, t)=
    \begin{cases}
      \bot &\text{if } x= \bot\,,\\
      0 &\text{if } 0\le x< 2\,,\\
      1 &\text{otherwise}\,.
    \end{cases}
  \end{gather*}
  Then
  \begin{align*}
    (( f\compop g)\compop h)( 0, 1) &= \bigvee_{ t_1+ t_2= 1} h\Big(
    \bigvee_{ t_3+ t_4= t_1} g( f( 0, t_3), t_4), t_2\Big) \\
    &= \bigvee_{ t_1\le 1} h\Big( \bigvee_{ t_3+ t_4= t_1} g( t_3,
    t_4), 0\Big)= \bigvee_{ t_1\le 1} h( 2 t_1, 0)= 1\,,
  \intertext{whereas}
    ( f\compop( g\compop h))( 0, 1) &= \bigvee_{ t_1+ t_2= 1\,}
    \bigvee_{\, t_3+ t_4= t_2} h( g( f( 0, t_1), t_3), t_4) \\
    &= \bigvee_{ t_1+ t_3\le 1} h( g( t_1, t_3), 0)= 0\,,
  \end{align*}
  as $g( t_1, t_3)< 2$ for $t_1+ t_3\le 1$.
\end{rem}

We will in Definition~\ref{de:rtef} below introduce a subclass
$\E\subseteq \F$ in which composition \emph{is} associative, and then
we will restrict ourselves to this subclass.  Until then, we take the
convention that $\compop$ binds to the right, that is, a composition
$f_1\compop f_2\compop f_3$ is to be read as $f_1\compop( f_2\compop
f_3)$.

Compositions of atomic RTEFs along paths are called \emph{linear}
RTEFs:

\begin{defi}
  A \emph{linear real-time energy function} is a finite composition
  $f_1 \compop f_2 \compop \dots \compop f_n$ of atomic RTEFs.
\end{defi}

\begin{exa}%
  \label{ex:linear}
  As an example, and also to show that linear RTEFs can have quite
  complex behavior, we show the linear RTEF associated to one of the
  paths in the satellite example of the introduction.  Consider the
  following (linear) RTEA:\@
  \begin{equation*}
    \tripleauto{0}{-20}{20}{2}{-20}{20}{5}{-10}{10}
  \end{equation*}
  Its linear RTEF $f$ can be computed as follows:
  \begin{equation*}
    f( x, t)=
    \begin{cases}
      \bot &\text{if } x< 20\text{ or }( 20\le x< 40\text{ and } x+ 2
      t< 44) \\
      &\phantom{\text{if } x< 20}\text{ or }( x\ge 40\text{ and } x+ 5
      t< 50) \\
      2.5 x+ 5 t- 110 &\text{if } 20\le x< 40\text{ and } x+ 2 t\ge 44
      \\
      x+ 5 t- 50 &\text{if } x\ge 40\text{ and } x+ 5 t\ge 50
    \end{cases}
  \end{equation*}

  We show a graphical representation of $f$ in Fig.~\ref{fig:graph}.
  The left part of the figure shows the boundary between two regions
  in the $(x,t)$ plane, corresponding to the minimal value $0$
  achieved by the function.  Below this boundary, no path exists
  through the corresponding RTEA\@.  Above, the function is linear in
  $x$ and $t$, as shown in the right part of the figure.  The
  coefficient of $t$ corresponds to the maximal rate in the RTEA;\@ the
  coefficient of $x$ depends on the relative position of $x$ with
  respect to (partial sums of) the bounds $b_i$.
\end{exa}

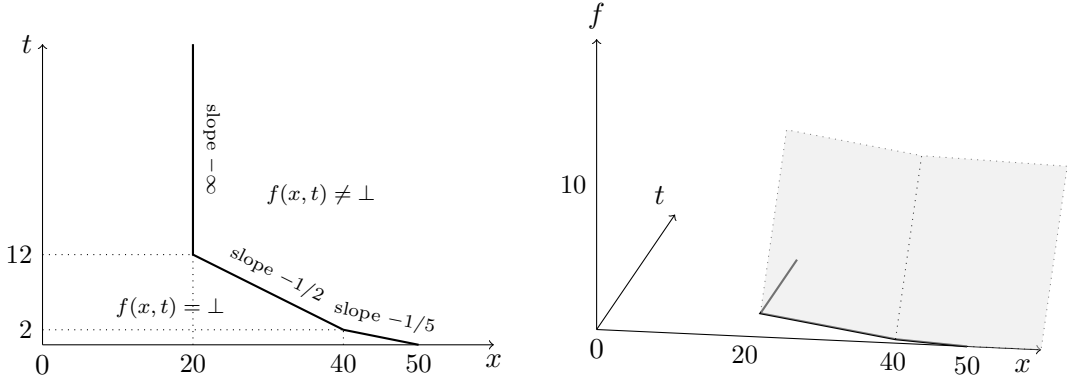
\begin{figure}[tbp]
\begin{tikzpicture}
\draw[->] (0,0) -- (6,0) node[anchor=north] {$x$};
\draw	(0,0) node[anchor=north] {0}
		(2,0) node[anchor=north] {\small $20$}
		(4,0) node[anchor=north] {\small $40$}
		(5,0) node[anchor=north] {\small $50$};
\draw[->] (0,0) -- (0,4) node[anchor=east] {$t$};
\draw	(0,0.2) node[anchor=east] {\small $2$}
	(0,1.2) node[anchor=east] {\small $12$};
\draw (1.7,0.5) node {\scriptsize $f(x,t)=\bot$};
\draw (3.7,2) node {\scriptsize $f(x,t)\neq\bot$};
\coordinate (x0) at (2,4);
\coordinate (x1) at (2,1.2);
\coordinate (x2) at (4,0.2);
\coordinate (xn) at (5,0);
\draw[thick] (x0)
-- node[anchor=south,sloped]{\tiny slope $-\infty$} (x1)
-- node[anchor=south,sloped]{\tiny slope $-1/2$} (x2)
-- node[anchor=south,sloped]{\tiny slope $-1/5$}   (xn);
\draw[dotted] (2,0) -- (x1);
\draw[dotted] (4,0) -- (x2);
\draw[dotted] (0,0.2) -- (x2);
\draw[dotted] (0,1.2) -- (x1);
\end{tikzpicture}
\tdplotsetmaincoords{75}{10}
\begin{tikzpicture}[tdplot_main_coords]
\draw[->] (0,0,0) -- (6,0,0) node[anchor=north east]{$x$};
\draw[->] (0,0,0) -- (0,6,0) node[anchor=south east]{$t$};
\draw[->] (0,0,0) -- (0,0,4) node[anchor=south]{$f$};
\draw	(0,0,0) node[anchor=north] {0}
		(2,0,0) node[anchor=north] {\small $20$}
		(4,0,0) node[anchor=north] {\small $40$}
		(5,0,0) node[anchor=north] {\small $50$}
                (0,0,2)   node[anchor=east] {\small $10$};

\coordinate (x0) at (2,4);
\coordinate (x1) at (2,1.2);
\coordinate (x2) at (4,0.2);
\coordinate (xn) at (5,0);
\draw[thick] (x0)
--  (x1)
--  (x2)
--  (xn);
\draw[dotted,fill=gray!20,opacity=0.5] (x1) -- ++(0,2,2) -- (4,2.2,2) -- (x2);
\draw[dotted,fill=gray!20,opacity=0.5] (x2) -- ++(0,2,2) -- (6,2,2) --
(6,0,0) -- (xn);
\end{tikzpicture}
\caption{Graphical representation of the linear RTEF from Example~\ref{ex:linear}.}\label{fig:graph}
\end{figure}

\subsection{Normal Form}

Next we need to see that all linear RTEFs can be converted to a
\emph{normal form}:

\begin{defi}%
  \label{de:normalform}
  A sequence $f_1,\dotsc, f_n$ of atomic RTEFs, with rates, bounds and
  prices $r_1,\dotsc, r_n$, $b_1,\dotsc, b_n$ and $p_1,\dotsc, p_n$,
  respectively, is in \emph{normal form} if
  \begin{itemize}
  \item $r_1<\dotsm< r_n$,
  \item $b_1\le\dotsm\le b_n$, and
  \item $p_1=\dotsm= p_{ n- 1}= 0$.
  \end{itemize}
\end{defi}

\begin{lem}%
  \label{le:normalform}
  For any linear RTEF $f$ there exists a sequence $f_1,\dotsc, f_n$ of
  atomic RTEFs in normal form such that
  $f= f_1\compop\dotsm\compop f_n$.
\end{lem}

\begin{exa}
  A normal form of the RTEF from Example~\ref{ex:linear} is as follows:
  \begin{equation*}
    \tripleauto{0}{0}{20}{2}{0}{40}{5}{-50}{50}
  \end{equation*}
  It is clear that its energy function is the same as the one of
  Example~\ref{ex:linear}: any run which satisfies the new constraints
  is equivalent to one which satisfies the old ones, and vice versa.
\end{exa}

\begin{proof}
  Let $f= f_1\compop\dotsm\compop f_n$, where $f_1,\dotsc, f_n$ are
  atomic RTEFs and assume $f_1, \dots, f_n$ is not in normal form.  If
  there is an index $k\in\{ 1,\dotsc, n- 1\}$ with $r_k\ge r_{ k+ 1}$,
  then we can use the following transformation to remove the state
  with rate $r_{ k+ 1}$:
  \begin{equation*}
    \label{eq:transshort}
    \doubleauto[r_{ k+ 2}]{r_k}{p_k}{b_k}{ r_{ k+ 1}}{ p_{ k+ 1}}{ b_{
        k+ 1}} \quad
    \underset{( r_k\ge r_{ k+ 1})}{\longmapsto} \simpleauto[r_{ k+
      2}]{r_k}{ p_k+ p_{ k+ 1}}{ \max( b_k, b_{ k+ 1}- p_k)}
  \end{equation*}

  Informally, any run through the RTEA for
  $f_1\compop\dotsm\compop f_n$ which maximizes output energy will
  spend no time in the state with rate $r_{ i+ 1}$, as this time may
  as well be spent in the state with rate $r_i$ without lowering
  output energy.  To make this argument precise, we prove that this
  transformation does not change the values of $f$.

  Let $f'$ denote the function which results from the transformation.
  Let $x\in L$ and $t\in \Realni$.  We show first that
  $f( x, t)\le f'( x, t)$, which is clear if $f( x, t)= \bot$.  If
  $f( x, t)\ne \bot$, then there is an accepting run through the RTEA
  corresponding to $f_1\compop\dotsm\compop f_n$.  Hence we have
  $t_1+\dotsm+ t_n= t$ such that
  $f( x, t)= x+ r_1 t_1+ p_1+\dotsm+ r_n t_n+ p_n$ and
  $x+\dotsm+ r_j t_j\ge b_j$ for all $j= 1,\dotsc, n$.  Let
  $t_k'= t_k+ t_{ k+ 1}$, $t_{ k+ 1}'= 0$, and $t_j'= t_j$ for all
  $j\notin\{ k, k+ 1\}$.  By $r_k\ge r_{ k+ 1}$, we know that
  $x+\dotsm+ r_k t_k'\ge b_k$ and
  $x+\dotsm+ r_{ k+ 1} t_{ k+ 1}'\ge b_{ k+ 1}$, hence
  $x+\dotsm+ r_j t_j'\ge b_j$ for all $j= 1,\dotsc, n$.  Hence this
  new run is also accepting, and $x+ r_1 t_1'+ p_1+\dotsm+ r_n t_n'+
  p_n\ge f( x, t)$.  Because $t_{ k+ 1}'= 0$, this also yields an
  accepting run through the RTEA for $f'$, showing that $f'( x, t)\ge
  f( x, t)$.

  The other inequality, $f( x, t)\ge f'( x, t)$, is clear if
  $f'( x, t)= \bot$.  Otherwise, there is an accepting run through the
  RTEA for $f'$.  Hence we have $t_1+\dotsm+ t_n= t$, with
  $t_{ k+ 1}= 0$, such that
  $f'( x, t)= x+ r_1 t_1+ p_1+\dotsm+ r_n t_n+ p_n$ and
  $x+\dotsm+ r_j t_j\ge b_j$ for all $j= 1,\dotsc, n$.  But then this
  is also an accepting run through the RTEA for $f$, showing that $f(
  x, t)\ge f'( x, t)$.

  We can hence assume that $r_1<\dotsm< r_n$.  To ensure the last two
  conditions of Definition~\ref{de:normalform}, we use the following
  transformation:
  \begin{equation*}
    \doubleauto[r_{ k+ 2}]{r_k}{p_k}{b_k}{ r_{ k+ 1}}{ p_{ k+ 1}}{ b_{
        k+ 1}} \quad
    \longmapsto \doubleauto[r_{ k+ 2}]{r_k}{0}{b_k}{ r_{ k+ 1}}{ p_k+
      p_{ k+ 1}}{ \max( b_k, b_{ k+ 1}- p_k)}
  \end{equation*}

  Informally, any run through the original RTEA can be copied to the
  other and vice versa, hence also this transformation does not change
  the values of $f$.  The precise argument is as follows.

  Let $f'$ denote the function which results from the transformation.
  Let $x\in L$ and $t\in \Realni$.  The inequality
  $f( x, t)\le f'( x, t)$ is again clear if $f( x, t)= \bot$, so
  assume otherwise.  Let $t_1+\dotsm+ t_n= t$ such that
  $f( x, t)= x+ r_1 t_1+ p_1+\dotsm+ r_n t_n+ p_n$ and
  $x+\dotsm+ r_j t_j\ge b_j$ for all $j= 1,\dotsc, n$.  Then this also
  yields an accepting run through the RTEA for $f'$, hence
  $f'( x, t)\ge f( x, t)$.  The proof that $f( x, t)\ge f'( x, t)$ is
  similar.
\end{proof}

Next we define a total preorder on normal-form sequences of atomic
RTEFs.  Using this ordering, we will later be able to show that the
semiring of general real-time energy functions is locally closed.

\begin{defi}
  Let $f_1,\dotsc, f_n$ and $f_1',\dotsc, f_{ n'}'$ be normal-form
  sequences of atomic RTEFs with rate sequences $r_1<\dotsm< r_n$ and
  $r_1'<\dotsm< r_{ n'}'$, respectively.  Then $f_1,\dotsc, f_n$ is
  \emph{not better than} $f_1',\dotsc, f_{ n'}'$, denoted
  $( f_1,\dotsc, f_n)\preceq( f_1',\dotsc, f_{ n'}')$, if
  $r_n\le r_{ n'}'$.
\end{defi}

Note that $( f_1,\dotsc, f_n)\preceq( f_1',\dotsc, f_{ n'}')$ does not
imply
$f_1\compop\dotsm\compop f_n\le f_1'\compop\dotsm\compop f_{ n'}'$
even for very simple functions.  For a counterexample, consider the
two following linear RTEFs $f= f_1$, $f'= f_1'\compop f_2'$ with
corresponding RTEAs:
\begin{gather*}
  f:\quad \simpleauto[]{4}{0}{0}
  \qquad\qquad
  f':\quad \doubleauto[]{1}{0}{1}{5}{0}{2}
\end{gather*}
We have $( f_1)\preceq( f_1', f_2')$, and for $x\ge 2$,
$f( x, t)= x+ 4t$ and $f'( x, t)= x+ 5t$, hence
$f( x, t)\le f'( x, t)$.  But $f( 0, 1)= 4$, whereas
$f'( 0, 1)= \bot$.

\begin{lem}%
  \label{le:subcommut}
  If $f= f_1\compop\dotsm\compop f_n$ and
  $f'= f_1'\compop\dotsm\compop f_{ n'}'$ are such that
  $( f_1,\dotsc, f_n)\preceq( f_1',\dotsc, f_{ n'}')$, then $f'\compop
  f\le f'$.
\end{lem}

Here the composition $f'\compop f$ is to be read as $f'\compop f=(
f_1\compop\dotsm\compop f_n)\compop( f_1'\compop\dotsm\compop f_{ n'}')$.

\begin{proof}
  Let $r_1<\dotsm< r_n$ and $r_1'<\dotsm< r_{ n'}'$ be the
  corresponding rate sequences, then $r_n\le r_{ n'}'$.
  The RTEAs for $f'\compop f$ and $f'$ are as follows, where we have
  transformed the former to normal form using that for all indices
  $i$, $r_i\le r_n\le r_{ n'}'$:
  \begin{equation*}
    \begin{tikzpicture}[shorten >=1pt, auto, ->, xscale=1]
      \begin{scope}
        \node at (-1.7,0) {$f'\compop f:$};
        \node[state, initial] (1') at (0,0) {$r_1'$};
        \node (2') at (1.8,0) {$\cdots$};
        \node[state] (n') at (3,0) {$r_{ n'}'$};
        \node[state, accepting] (end) at (6.8,0) {};
        \path (1') edge node[below] {$b_1'$} node[above] {$0$} (2');
        \path (2') edge (n');
        \path (n') edge node[below] {$\max( b_{ n'}', b_n- p_{ n'}')$}
        node[above] {$p_n+ p_{ n'}'$} (end);
      \end{scope}
      \begin{scope}[yshift=-4em]
        \node at (-1.6,0) {$f':$};
        \node[state, initial] (1') at (0,0) {$r_1'$};
        \node (2') at (1.8,0) {$\cdots$};
        \node[state] (n') at (3,0) {$r_{ n'}'$};
        \node[state, accepting] (end) at (6.8,0) {};
        \path (1') edge node[below] {$b_1'$} node[above] {$0$} (2');
        \path (2') edge (n');
        \path (n') edge node[below] {$b_{ n'}'$}
        node[above] {$p_{ n'}'$} (end);
      \end{scope}
    \end{tikzpicture}
  \end{equation*}
  As $p_n+ p_{ n'}'\le p_{ n'}'$ (because $p_n\le 0$) and
  $\max( b_{ n'}', b_n- p_{ n'}')\ge b_{ n'}'$, it is clear that
  $f'\compop f( x, t)\le f'( x, t)$ for all $x\in L$, $t\in \Realni$.
\end{proof}

\subsection{General Real-Time Energy Functions}

We now consider all paths that may arise in a real-time energy
automaton.  When two locations of an automaton may be joined by two
distinct paths, the optimal output energy is naturally obtained by
taking the maximum over both paths.  This gives rise to the following
definition.

\begin{defi}
  Let $f, g\in \F$. The function $f\vee g$ is defined as the
  pointwise supremum:
  \begin{equation*}
    \forall t\in \Realni: ( f\vee g)( t)= f( t)\vee g( t)
  \end{equation*}
\end{defi}

Let $\one, \bot, \top\in \F$ be the functions defined by
$\one( x, t)= x$ and $\bot( x, t)= \bot$ for all
$x\in L, t\in[ 0, \infty]$, $\top( \bot, t)= \bot$, and
$\top( x, t)= \infty$ for all $x, t\in[ 0, \infty]$.

\begin{lem}
  With operation $\vee$, $\F$ forms a complete lattice with bottom and
  top elements $\bot$ and $\top$.
\end{lem}

\begin{proof}
  For $\G\subseteq \F$, the supremum $\bigvee \G$ is given pointwise
  as $( \bigvee \G)( t)= \bigvee_{ f\in \G} f( t)$.  Completeness of
  $\Realni$ thus implies completeness of $\F$.  The claim regarding
  $\bot$ and $\top$ is clear.
\end{proof}

Finally, a cycle in an RTEA results in a $^*$-operation:

\begin{defi}
  Let $f \in \F$. The Kleene star of $f$ is the
  function $f^*\in \F$ such that
  \begin{equation*}
    \forall t\in \Realni: f^*( t)= \bigvee_{n\ge 0} f^n( t)
  \end{equation*}
\end{defi}

Note that $f^*$ is defined for all $f\in \F$ because $\F$ is a
complete lattice.  We can now define the set of general real-time
energy functions, corresponding to general RTEAs:

\begin{defi}%
  \label{de:rtef}
  The set $\E$ of \emph{real-time energy functions} is the subset of
  $\F$ generated by atomic RTEFs and $\{ \bot, \top\}$, \ie~the subset
  of $\F$ inductively defined by
  \begin{itemize}
  \item $\A\cup\{ \bot, \top\}\subseteq \E$,
  \item if $f, g\in \E$, then $f\compop g\in \E$ and $f\vee g\in \E$.
  \end{itemize}
\end{defi}

\noindent
We will show below that $\E$ is locally closed, which entails that for
each $f\in \E$, also $f^*\in \E$, hence $\E$ indeed encompasses all
RTEFs.

\begin{defi}%
  \label{de:pwl}
  A a function $f\in \F$ is \emph{piecewise linear} (PWL) if there
  exists a finite covering of disjoint convex polyhedra
  $X_1,\dotsc, X_N\subseteq L\times \Realni$, \ie~such that
  $X_1\cup\dotsm\cup X_N= L\times \Realni$ and
  $X_i\cap X_j= \emptyset$ for $i\ne j$, and functions
  $f_1,\dotsc, f_N\in \F$ such that for every $i$, $f_i$ is an affine
  function on $X_i$ and equal to $\bot$ outside, and
  $f= \bigvee_{ i= 1}^N f_i$.
\end{defi}

Also recall that a function $f\in \F$ is \emph{right-continuous} if
$f\big( \bigwedge X\big)= \bigwedge_{( x, t)\in X} f( x, t)$ for all
subsets $X\subseteq L\times \Realni$.

\begin{lem}%
  \label{le:pwl}
  All functions in $\E$ are PWL and right-continuous.
\end{lem}

\begin{proof}
  It is clear that all atomic RTEFS and also $\bot$ and $\top$ are PWL
  and right-continuous.  We proceed by structural induction.  Let $f,
  g\in \E$ and assume $f$ and $g$ to be PWL and right-continuous.

  Let first $h= f\vee g$.  To show that $h$ is right-continuous, let
  $X\subseteq L\times \Realni$, then
  \begin{align*}
    h\big( \bigwedge X\big) = f\big( \bigwedge X\big)\vee g\big(
    \bigwedge X\big) &= \Big( \bigwedge_{( x, t)\in X} f( x,
    t)\Big)\vee\Big( \bigwedge_{( x, t)\in X} g( x, t)\Big) \\
    &= \bigwedge_{( x, t)\in X}( f( x, t)\vee g( x, t)) = \bigwedge_{(
      x, t)\in X} h( x, t)\,.
  \end{align*}

  To show that $h$ is PWL, take coverings
  $X_1\cup\dotsm\cup X_N= Y_1\cup\dotsm\cup Y_M= L\times \Realni$ such
  that $f= \bigvee_{ i= 1}^N f_i$, $g= \bigvee_{ i= 1}^M g_i$, as in
  Def.~\ref{de:pwl}.  Let $Z_{ i j}= X_i\cap Y_j$ and
  $h_{ i j}= f_i\vee g_j$ for $i= 1,\dotsc, N$, $j= 1,\dotsc, M$, then
  $Z= Z_{ 1 1}\cup\dotsm\cup Z_{ N M}$ is a finite convex covering of
  $L\times \Realni$, each $h_{ i j}$ is affine on $Z_{ i j}$ and equal
  to $\bot$ outside, and $h= \bigvee h_{ i j}$.

  Now let $h= f\compop g$.  Let again
  $X_1\cup\dotsm\cup X_N= Y_1\cup\dotsm\cup Y_M= L\times \Realni$ be
  coverings such that $f= \bigvee_{ i= 1}^N f_i$,
  $g= \bigvee_{ i= 1}^M g_i$.  For every $i= 1,\dotsc, N$,
  $j= 1,\dotsc, M$, define sets
  $Z_{ i j}\subseteq L\times[ 0, \infty]\times[ 0, \infty]$ by
  \begin{equation*}
    Z_{ i j} = \{( x, t_1, t_2)\mid( f_i( x, t_1), t_2)\in Y_j\}
  \end{equation*}
  Being inverse images of convex polyhedra by linear functions, every
  $Z_{ i j}$ is itself a convex polyhedron; also, the $Z_{ i j}$ are
  disjoint.

  We have
  \begin{align*}
    h( x, t) = \bigvee_{ t_1+ t_2= t} g( f( x, t_1), t_2)
    &= \bigvee_{ t_1+ t_2= t} g\Big( \bigvee_{ i= 1}^N f_i( x, t_1), t_2\Big)
    \\
    &= \bigvee_{ t_1+ t_2= t} \bigvee_{ i= 1}^N g( f_i( x, t_1), t_2)
    = \bigvee_{ t_1+ t_2= t} \bigvee_{ i= 1}^N \bigvee_{ j= 1}^M g_j(
    f_i( x, t_1), t_2)
  \end{align*}

  Define functions
  $h_{ i j}: L\times[ 0, \infty]\times[ 0, \infty]\to L$, for every
  $i= 1,\dotsc, N$, $j= 1,\dotsc, M$, by
  $h_{ i j}( x, t_1, t_2)= g_j( f_i( x, t_1), t_2)$.  By definition,
  for every $i, j$, $h_{ i j}$ is affine on $Z_{ i j}$ and equal to
  $\bot$ outside.

  Continuing the equalities from above,
  \begin{equation*}
    h( x, t) = \bigvee_{ t_1+ t_2= t} \bigvee_{ i= 1}^N \bigvee_{ j= 1}^M
    h_{ i j}( x, t_1, t_2) = \bigvee_{ i= 1}^N \bigvee_{ j= 1}^M
    \bigvee_{ t_1+ t_2= t} h_{ i j}( x, t_1, t_2)
  \end{equation*}
  which holds because the $h_{ i j}$ are defined on disjoint sets.

  Now fix $i$ and $j$
  and define $\hat h_{ i j}\in \F$ by
  $\hat h_{ i j}( x, t)= \bigvee_{ t_1+ t_2= t} h_{ i j}( x, t_1,
  t_2)$.
  The function $\hat h_{ i j}$ is obtained from $h_{ i j}$ by
  ``sweeping'' $Z_{ i j}$ with the planes $t_1+ t_2= t$.  Now split
  $Z_{ i j}$ into pieces according to where this sweep meats its
  vertices, then $Z_{ i j}= \bigcup_{ k= 1}^L Z_{ i j k}$ for some
  $L\in \Nat$.  This creates a finite split of $Z_{ i j}$ into
  disjoint convex polyhedra.

  Split $\hat h_{ i j}$ into similar pieces $\hat h_{ i j k}$ such
  that each $\hat h_{ i j k}$ is affine on $Z_{ i j k}$ and equal to
  $\bot$ outside, then $\hat h_{ i j}= \bigvee_{ k= 1}^L \hat h_{ i j
    k}$.  Let $\hat Z_{ i j k}=\{( x, t_1+ t_2)\mid( x, t_1, t_2)\in
  Z_{ i j k}\}$, then each $\hat Z_{ i j k}$ is a convex polyhedron,
  and the $\hat Z_{ i j k}$ define a partition of $L\times[ 0,
  \infty]$.

  Continuing the equalities from above,
  \begin{equation*}
    h( x, t) = \bigvee_{ i= 1}^N \bigvee_{ j= 1}^M \hat h_{ i j}( x, t)
    = \bigvee_{ i= 1}^N \bigvee_{ j= 1}^M \bigvee_{ k= 1}^L \hat h_{ i
      j k}( x, t)
  \end{equation*}
  where each $\hat h_{ i j k}$ is affine on $\hat Z_{ i j k}$ and
  equal to $\bot$ outside.  That is, $h$ is PWL\@.

  To show that $h$ is right-continuous, first note that $f$ and $g$
  being right-continuous implies that we can assume that in the
  coverings
  $X_1\cup\dotsm\cup X_N= Y_1\cup\dotsm\cup Y_M= L\times \Realni$,
  each $X_i$ and each $Y_j$ include their lower boundaries.  That is,
  for every $i= 1,\dotsc, N$, $j= 1,\dotsc, M$, $X\subseteq X_i$, and
  $Y\subseteq Y_j$, also $\bigwedge X\in X_i$ and $\bigwedge Y\in
  Y_j$.

  Next we show that the sets $Z_{ i j}$ have the same property.  Let
  $Z\subseteq Z_{ i j}$ and $\bigwedge Z=( z, u_1, u_2)$.  Let
  $Y=\{( f_i( x, t_1), t_2)\mid( x, t_1, t_2)\in Z\}\subseteq Y_j$,
  then by right-continuity of $g_j$, we have $\bigwedge Y\in Y_j$.
  Now by right-continuity of $f_i$,
  $\bigwedge Y= \bigwedge\{( f_i( \bigwedge\{( x, t_1)\}), t_2)\}=(
  f_i( z, u_1), u_2)$, hence $( z, u_1, u_2)\in Z_{ i j}$.

  Hence all functions $h_{ i j}$ are right-continuous, and per their
  definition, this also applies to all functions $\hat h_{ i j}$.
  This means that we can assume all subdivisions $Z_{ i j k}$ to
  include their lower boundaries, and then the polyhedra $\hat Z_{ i j
    k}\subseteq L\times \Realni$ have the same property.  That is to
  say, $h$ is right-continuous.
\end{proof}

\begin{lem}%
  \label{le:compopassoc}
  On $\E$, the operation $\compop$ is associative.
\end{lem}

Note that Remark~\ref{re:notassoc} does not apply, because the
function $g$ in that example is not right-continuous.  On the other
hand, the proof uses both right-continuity and piecewise linearity.

\begin{proof}
  Let $f, g, k\in \E$; we prove that
  $( f\compop g)\compop k= f\compop( g\compop k)$.  Unrolling the
  definition, we see that we need to show that for all $x\in L$,
  $t, t_3\in \Realni$,
  $k( \bigvee_{ t_1+ t_2= t} g( f( x, t_1), t_2), t_3)= \bigvee_{ t_1+
    t_2= t} k( g( f( x, t_1), t_2), t_3)$.

  Let $X_1\cup\dotsm\cup X_N= Y_1\cup\dotsm\cup Y_M= L\times \Realni$,
  $f= \bigvee_{ i= 1}^N f_i$, $g= \bigvee_{ i= 1}^M g_i$, $Z_{ i j}$,
  and $h_{ i j}$, with
  $h_{ i j}( x, t_1, t_2)= g_j( f_i( x, t_1), t_2)$ like in the proof
  of Lemma~\ref{le:pwl}.  As $k$ is PWL, the above equality reduces to
  $k( \bigvee_{ t_1+ t_2= t} h_{ i j}( x, t_1, t_2), t_3)= \bigvee_{
    t_1+ t_2= t} k( h_{ i j}( x, t_1, t_2), t_3)$.
  We know that $Z_{ i j}$ includes its lower boundary, and by
  linearity of $h_{ i j}$, the value
  $\bigvee_{ t_1+ t_2= t} h_{ i j}( x, t_1, t_2)$ is assumed on that
  lower boundary.  The equality follows by piecewise linearity of $k$.
\end{proof}

\begin{prop}%
  \label{le:Fsemiring}
  With operations $\vee$ and $\compop$, $\E$ forms an idempotent
  semiring, with $\bot$ as unit for $\vee$ and $\one$ as unit for
  $\compop$.
\end{prop}


\begin{proof}
  The operation $\vee$ is clearly associative, and $\compop$ is so by
  Lemma~\ref{le:compopassoc}.

  Let $f\in \E$.  It is clear that $f\vee \bot= \bot\vee f= f$ and
  $f\compop \bot= \bot\compop f= \bot$.  For $f\compop \one$ and
  $\one\compop f$, we have
  $\comp{f}{\one}{t}(x)= \bigvee_{ t_1+ t_2= t} \one( f( x, t_1),
  t_2)= \bigvee_{ t_1+ t_2= t} f( x, t_1)= f( x, t)$
  because of monotonicity of $f$.  Similarly,
  $\comp{\one}{f}{t}(x)= \bigvee_{ t_1+ t_2= t} f( \one( x, t_1),
  t_2)= \bigvee_{ t_1+ t_2= t} f( x, t_2)= f( x, t)$
  because of monotonicity of $f$.

  We only miss to show the distributive laws.  Let $f, g, h\in \E$ and
  $t\in \Realni$, then
  \begin{align*}
    ( f\compop( g\vee h))( t)
    &= \bigvee_{ t_1+ t_2= t}( g\vee h)( t_2)\circ f( t_1) \\
    &= \bigvee_{ t_1+ t_2= t}( g( t_2)\vee h( t_2))\circ f( t_1) \\
    &= \bigvee_{ t_1+ t_2= t} g( t_2)\circ f( t_1)\vee h( t_2)\circ f(
      t_1) \\
    &= \bigvee_{ t_1+ t_2= t} g( t_2)\circ f( t_1)\vee \bigvee_{ t_1+
      t_2= t} h( t_2)\circ f( t_1) \\
    &= f\compop g( t)\vee f\compop h( t)=( f\compop g\vee f\compop h)( t)\,.
  \end{align*}

  Similarly, and using monotonicity of $h$, we see that
  \begin{align*}
    (( f\vee g)\compop h)( t)
    &= \bigvee_{ t_1+ t_2= t} h( t_2)\circ( f\vee g)( t_1) \\
    &= \bigvee_{ t_1+ t_2= t} h( t_2)\circ( f( t_1)\vee g( t_1)) \\
    &= \bigvee_{ t_1+ t_2= t} h( t_2)\circ f( t_1)\vee h( t_2)\circ g(
      t_1) \\
    &= \bigvee_{ t_1+ t_2= t} h( t_2)\circ f( t_1)\vee \bigvee_{ t_1+
      t_2= t} h( t_2)\circ g( t_1) \\
    &= f\compop h( t)\vee g\compop h( t)=( f\compop h\vee
      g\compop h)( t)\,.
  \end{align*}
  The proof is complete.
\end{proof}

\begin{lem}%
  \label{le:locallyclosed}
  For every $f\in \E$ there exists $N\ge 0$ so that
  $f^*= \bigvee_{ n= 0}^N f^n$.
\end{lem}

\begin{proof}
  By distributivity, we can write $f$ as a finite supremum
  $f= \bigvee_{ k= 1}^m f_k$ of linear energy functions
  $f_1,\dotsc, f_m$.  For each $k= 1,\dotsc, m$, let
  $f_k= f_{ k, 1}\compop\dotsm\compop f_{ k, n_k}$ be a normal-form
  representation.  By re-ordering the $f_k$ if necessary, and because
  $\preceq$ is total, we can assume that
  $( f_{ k, 1},\dotsc, f_{ k, n_k})\preceq( f_{ k+ 1, 1},\dotsc, f_{
    k+ 1, n_{ k+ 1}})$ for every $k= 1,\dotsc, n- 1$.

  We first show that
  $f^*\le \bigvee_{ 0\le n_1,\dotsc, n_m\le 1} f_1^{
    n_1}\compop\dotsm\compop f_m^{ n_m}$:
  The expansion of $f^*={( \bigvee_{ k= 1}^m f_k)}^*$ is an infinite
  supremum of finite compositions
  $f_{ i_1}\compop\dotsm\compop f_{ i_p}$.  By
  Lemma~\ref{le:subcommut}, any occurrence of
  $f_{ i_j}\compop f_{ i_{ j+ 1}}$ in such compositions with
  $i_j\ge i_{ j+ 1}$ can be replaced by $f_{ i_{ j+ 1}}$.  The
  compositions which are left have $i_j< i_{ j+ 1}$ for every $j$, so
  the claim follows.

  Now
  $\bigvee_{ 0\le n_1,\dotsc, n_m\le 1} f_1^{ n_1}\compop\dotsm\compop
  f_m^{ n_m}\le \bigvee_{ n= 0}^m {\big( \bigvee_{ k= 1}^m f_k\big)}^n=
  \bigvee_{ n= 0}^m f^n\le f^*$, so with $N= m$ the proof is complete.
\end{proof}

\begin{cor}
  $\E$ is locally closed, hence a $^*$-continuous Kleene algebra.
\end{cor}

\begin{proof}
  For every $f\in \E$ there is $N\ge 0$ so that
  $f^*= \bigvee_{ n= 0}^N f^n$ (Lemma~\ref{le:locallyclosed}), hence
  $\bigvee_{ n= 0}^N f^n= \bigvee_{ n= 0}^{ N+ 1} f^n$.  Thus $\E$ is
  locally closed, and by Lemma~\ref{le:starcont}, a $^*$-continuous
  Kleene algebra.
\end{proof}

\begin{exa}%
  \label{ex:star}
  To illustrate, we compute the Kleene star of the supremum
  $f= f_1\vee f_2$ of two linear RTEFs as below.  These are slight
  modifications of some RTEFs from the satellite example, modified to
  make the example more interesting:
  \begin{equation*}
    \begin{tikzpicture}[baseline, shorten >=1pt, ->, yscale=.4,
      xscale=1.2]
      \node at (7,-2) {$f_1$};
      \node at (7,2) {$f_2$};
      \node[state, initial] (0) at (0,0) {$0$};
      \node[state, accepting] (3) at (6,0) {};
      \node (f0) at (0,-2) {};
      \node (f1) at (2,-2) {};
      \node[state] (f2) at (4,-2) {$4$};
      \node (f3) at (6,-2) {};
      \node (h0) at (0,2) {};
      \node[state] (h1) at (2,2) {$1$};
      \node[state] (h2) at (4,2) {$5$};
      \node (h3) at (6,2) {};
      \draw (0) .. controls (f0) .. node[above, pos=.8] {$0$}
      node[below, pos=.8] {$30$} (f2);
      \draw (f2) .. controls (f3) .. node[above, pos=.17] {$-10$}
      node[below, pos=.17] {$30$} (3);
      \draw (0) .. controls (h0) .. node[above, pos=.83] {$0$}
      node[below, pos=.83] {$20$} (h1);
      \path (h1) edge node[above] {$0$} node[below] {$40$} (h2);
      \draw (h2) .. controls (h3) .. node[above, pos=.17] {$-50$}
      node[below, pos=.17] {$50$} (3);
    \end{tikzpicture}
  \end{equation*}

  These functions are in normal form and $f_1\preceq f_2$.
  Lemma~\ref{le:locallyclosed} and its proof allow us to conclude that
  $f^*= \one\vee f_1\vee f_2\vee f_1\compop f_2$.
  Figure~\ref{fi:star} shows the boundaries of definition of these
  functions and the regions in the $(x,t)$ plane where each of them
  dominates the supremum.
\end{exa}

\begin{figure}[tbp]
  \centering
\begin{tikzpicture}[xscale=1.5, yscale=1]
\definecolor{darkteal}{rgb}{0,0.3,0.3};
\definecolor{darkblue}{rgb}{0,0,0.5};
\definecolor{darkred}{rgb}{0.5,0,0};
\draw[->] (0,0) -- (8,0) node[anchor=north] {$x$};
\draw	(0,0) node[anchor=north] {0}
		(2,0) node[anchor=north] {\small $20$}
		(3,0) node[anchor=north] {\small $30$}
		(3.66,0) node[anchor=north] {\small $37$}
		(4,0) node[anchor=north] {\small $40$}
		(5,0) node[anchor=north] {\small $50$}
		(6,0) node[anchor=north] {\small $60$}
;
\draw[->] (0,0) -- (0,8.5) node[anchor=east] {$t$};
\draw
        (0,0.2) node[anchor=east] {\small $2$}
        (0,0.7) node[anchor=east] {\small $7$}
        (0,2.2) node[anchor=east] {\small $22$}
        (0,4) node[anchor=east] {\small $40$}
	(0,5.28) node[anchor=east] {\small $53$}
	(0,5.55) node[anchor=east] {\small $55$}
        (0,8) node[anchor=east] {\small $80$}
;
\coordinate (f0) at (2.995,8);
\coordinate (f1) at (2.995,0);
\coordinate (h0) at (2,8);
\coordinate (h1) at (2,2.2);
\coordinate (h2) at (4,0.2);
\coordinate (h3) at (5,0);
\coordinate (fh0) at (3.01,8);
\coordinate (fh1) at (3.01,0.7);
\coordinate (fh2) at (5,0.2);
\coordinate (fh3) at (6,0);

\fill[pattern=my north west lines, line space=8pt, pattern color=red] (h0) -- (2,2.6) -- (3,1.8) -- (f0);
\fill[pattern=my north west lines, line space=8pt, pattern color=red] (3.66,8) -- (3.66,5.33) -- (4,4) -- (8,4)
-- (8,8);
\fill[pattern=my dots, pattern color=teal] (3.66,8) -- (3.66,5.33) -- (3,5.5) -- (f0);
\fill[pattern=my north east lines, line space=8pt, pattern color=blue] (8,0.25) -- (3,0.25) -- (3,5.5) --
(3.66,5.33) -- (4,4) -- (8,4);
\fill[fill=black,opacity=0.05] (0,8) -- (0,0) -- (8,0) --(8,0.25) --
(3,0.25) -- (3,1.8) -- (2,2.6) -- (2,8);

\draw[thick,color=blue] (f1) -- node[anchor=north,sloped,color=darkblue,fill=white,pos=.4] {\scriptsize
  $f_1(x,t)= 4t+x-10$} (f0);
\draw[thick,color=red] (h0) -- (h1) -- (h2) -- (h3);
\draw[thick,color=teal] (fh0) -- (fh1) -- (fh2) -- (fh3);

\draw[dotted] (3.66,8) -- (3.66,5.33) -- (4,4) -- (8,4);
\draw[dotted] (3.66,5.33) -- (3,5.5);
\draw[dotted] (2,2.6) -- (3,1.8);
\draw[dotted] (3,0.25) -- (8,0.25);

\draw[thick,dashed,color=red] (h2) -- node[anchor=south east,sloped,color=darkred,fill=white,pos=.47] {\scriptsize
  $f_2(x,t)= 5t+5x-210$} node[anchor=north east,sloped,color=darkred,fill=white,pos=.45] {\scriptsize
  $f_2(x,t)= 5t+x-50$} ++(0,7.8);
\draw[thick,dashed,color=teal] (fh2) -- node[anchor=south east,sloped,color=darkteal,fill=white,pos=.57] {\scriptsize
  $\comp{f_1}{f_2}{x,t}= 5t+1.25x-72.5$} node[anchor=north east,sloped,color=darkteal,fill=white,pos=.48] {\scriptsize
  $\comp{f_1}{f_2}{x,t}= 5t+x-60$} ++(0,7.8);

\draw (1.7,0.5) node[fill=white] {\small $f^*=\one$};
\draw (6.5,6) node[fill=white] {\small $f^*=f_2$};
\draw (6.5,2.2) node[fill=white] {\small $f^*=f_1$};
\draw (3.35,6.65) node[rotate=90,fill=white] {\small $f^*=f_1 \compop f_2$};
\end{tikzpicture}
  \caption{Computation of $f^*$ from Example~\ref{ex:star}.}\label{fi:star}
\end{figure}
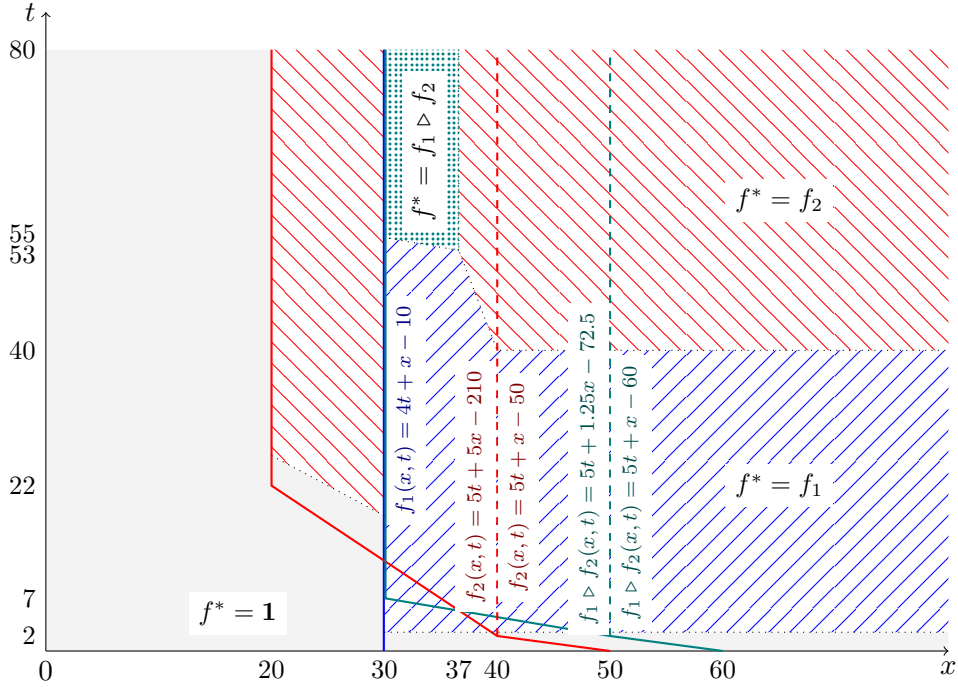

\subsection{Infinite Products}

Let $\Bool=\{ \lfalse, \ltrue\}$ denote the Boolean lattice with
standard order $\lfalse< \ltrue$.  Let $\V$ denote the set of
monotonic functions $v: L\times[ 0, \infty]\to \Bool$ for which
$v( \bot, t)= \lfalse$ for all $t\in[ 0, \infty]$.  We define an
infinite product operation $\E^\omega\to \V$:

\begin{defi}
  For an infinite sequence of functions $f_0, f_1,\dotsc\in \E$,
  $\prod_{ n\ge 0} f_n\in \V$ is the function defined for $x\in L$,
  $t\in[ 0, \infty]$ by $\prod_{ n\ge 0} f_n( x, t)= \ltrue$ iff there
  is an infinite sequence $t_0, t_1,\dotsc\in \Realni$ such that
  $\sum_{ n= 0}^\infty t_n= t$ and for all $n\ge 0$,
  $f_n( t_n)\circ\dotsm\circ f_0( t_0)( x)\ne \bot$.
\end{defi}

Hence $\prod_{ n\ge 0} f_n( x, t)= \ltrue$ iff in the infinite
composition $f_0\compop f_1\compop\dotsm( x, t)$, all finite prefixes
have values $\ne \bot$.
There is a (left) action of $\E$ on $\V$ given by
$( f, v)\mapsto f\compop v$, where the composition $f\compop v$ is
given by the same formula as composition $\compop$ on $\F$.  Let
$\bot\in \V$ denote the function given by $\bot( x, t)= \lfalse$.

\begin{lem}%
  \label{le:FVsemimod}
  With the $\E$-action $\compop$, $\vee$ as addition, and $\bot$ as
  unit, $\V$ is an idempotent left $\E$-semimodule.
\end{lem}

\begin{proof}
  Similar to the proof of Proposition~\ref{le:Fsemiring}.
\end{proof}

Let $\U\subseteq \V$ be the $\E$-subsemimodule generated by $\E$, that
is, the smallest (idempotent left) $\E$-semimodule contained in $\V$
which contains all infinite products of functions in~$\E$.

\begin{prop}%
  \label{pr:EUsckoa}
  $( \E, \U)$ forms a $^*$-continuous Kleene $\omega$-algebra.
\end{prop}

\begin{proof}
  We first show that $( \E, \U)$ forms a generalized $^*$-continuous
  Kleene algebra: Let $f, g\in \E$ and $v\in \U$, then we need to see
  that
  $f\compop g^*\compop v= \bigvee_{ n\ge 0} f\compop g^n\compop v$.
  The right-hand side is trivially less than or equal to the left-hand
  side.  For the other inequality, as $g$ is $^*$-closed, we have
  $N\ge 0$ such that $g^*= \bigvee_{ n= 0}^N g^n$, and then
  \begin{equation*}
    f\compop g^*\compop v= f\compop\big( \bigvee_{ n= 0}^N
    g^n\big)\compop v= \bigvee_{ n= 0}^N f\compop g^n\compop v\le
    \bigvee_{ n\ge 0} f\compop g^n\compop v\,.
  \end{equation*}

  We now need to show that $( \E, \U)$ satisfies the conditions
  \Ax1--\Ax4 in Section~\ref{se:starcontkleom}.  As to \Ax1, let $f_0,
  f_1,\dotsc\in \E$, $x\in L$, and $t\in[ 0, \infty]$.  Then
  \begin{align*}
    f_0\compop \prod_{ n\ge 0} f_{ n+ 1}( x, t)
    &= \bigvee_{ t_0+ t'= t} \prod_{ n\ge 0} f_{ n+ 1}( t')\circ f_0(
      t_0)( x) \\
    &= \ltrue\text{ iff } \exists t_0+ t'= t: \prod_{ n\ge 0} f_{ n+
      1}( t')\circ f_0( t_0)( x)= \ltrue \\
    &= \ltrue\text{ iff } \exists t_0+ t'= t: \exists t_1+ t_2+\dotsm=
      t': \forall n\ge 1: \\
    &\hspace*{15em} f_n( t_n)\circ\dotsm\circ f_0( t_0)\ne \bot \\
    &= \prod_{ n\ge 0} f_n( x, t)\,.
  \end{align*}

  For \Ax2, let $f_0, f_1,\dotsc\in \E$, $x\in L$, $t\in[ 0, \infty]$,
  and $0= n_0\le n_1\le\dotsm$ a sequence which increases without a
  bound.  Then
  \begin{align*}
    & \smash[b]{\prod_{ k\ge 0}( f_{
      n_k}\compop\dotsm\compop f_{ n_{ k+ 1}- 1})( x, t)}= \ltrue \\
    &\qquad\qquad\text{iff } \exists u_0+ u_1+\dotsm= t: \forall k\ge 0: \\
    &\hspace*{10em} ( f_{ n_k}\compop\dotsm\compop f_{ n_{ k+ 1}- 1})(
      u_k)\circ\dotsm\circ( f_0\compop\dotsm\compop f_{ n_1- 1})(
      u_0)( x)\ne \bot \\
    &\qquad\qquad\text{iff } \exists u_0+ u_1+\dotsm= t: \forall k\ge 0: \exists
      t^k_0,\dotsc, t^k_{ n_{ k+ 1}- 1}: \\
    &\hspace*{10em} t^k_0+\dotsm+ t^k_{ n_1- 1}= u_0,\dotsc, t^k_{
      n_k}+\dotsm+ t^k_{ n_{ k+ 1}- 1}= u_k, \\
    &\hspace*{10em} f_{ n_{ k+ 1}- 1}( t^k_{ n_{
      k+ 1}- 1})\circ\dotsm\circ f_0( t^k_0)\ne \bot\,.
  \end{align*}
  We can use a diagonal-type argument to finish the
  proof: For every $k$, we have
  $t^{ k+ 1}_0,\dotsc, t^{ k+ 1}_{ n_{ k+ 2}- 1}$ such that
  $f_{ n_{ k+ 2}- 1}( t^{ k+ 1}_{ n_{ k+ 2}- 1})\circ\dotsm\circ f_0(
  t^{ k+ 1}_0)\ne \bot$.
  But then also
  $f_{ n_{ k+ 1}- 1}( t^{ k+ 1}_{ n_{ k+ 1}- 1})\circ\dotsm\circ f_0(
  t^{ k+ 1}_0)\ne \bot$,
  hence we can update
  $t^k_0:= t^{ k+ 1}_0,\dotsc, t^k_{ n_{ k+ 1}- 1}:= t^{ k+ 1}_{ n_{
      k+ 1}- 1}$.
  In the limit, we have $t_0, t_1,\dotsc$ with
  $t_0+\dotsm+ t_{ n_1- 1}= u_0, \dotsc$, hence $t_0+ t_1+\dotsm= t$,
  and $f_n( t_n)\circ\dotsm\circ f_0( t_0)= \bot$.


  To show the third condition, we prove that for all
  $f_0, f_1,\dotsc, g_0, g_1,\dotsc\in \E$,
  \begin{equation}
    \label{eq:ax3}
    \prod_{ n\ge 0}( f_n\vee g_n)= \adjustlimits \bigvee_{ h_n\in\{
      f_n, g_n\}\;} \prod_{ n\ge 0} h_n\,,
  \end{equation}
  which implies \Ax3.  By monotonicity of the infinite product, the
  right-hand side is less than or equal to the left-hand side.  To
  show the other inequality, let $x\in L$ and $t\in[ 0, \infty]$ and
  suppose that $\prod_{ n\ge 0}( f_n\vee g_n)( x, t)= \ltrue$.  We
  show that there is a choice of functions $h_n\in\{ f_n, g_n\}$ for
  all $n\ge 0$ such that $\prod_{ n\ge 0} h_n( x, t)= \ltrue$.

  Consider the infinite ordered binary tree where each node at level
  $n\ge 0$ is the source of an edge labeled $f_n$ and an edge labeled
  $g_n$, ordered as indicated.  We can assign to each node $u$ the
  composition $h_u$ of the functions that occur as the labels of the
  edges along the unique path from the root to that node.

  Let us mark a node $u$ if $h_u( x, t)\ne \bot$.  As
  $\prod_{ n\ge 0}( f_n\vee g_n)( x, t)= \ltrue$, each level contains
  a marked node.  Moreover, whenever a node is marked and has a
  predecessor, its predecessor is also marked.  By K{\"o}nig's
  lemma~\cite{Konig27} there is an infinite path going through marked
  nodes.  This infinite path gives rise to the sequence
  $h_0, h_1,\dotsc$ with $\prod_{ n\ge 0} h_n( x, t)= \ltrue$.

  For \Ax4, we need to see that for all $f, g_0, g_1,\dotsc\in \E$,
  \begin{equation*}
    \prod_{ n\ge 0} f^*\compop g_n= \adjustlimits \bigvee_{ k_0,
      k_1,\dotsc\ge 0\;} \prod_{ n\ge 0} f^{k_n}\compop g_n\,.
  \end{equation*}
  Again the right-hand side is less than or equal to the left-hand
  side because of monotonicity of the infinite product.  To show the
  other inequality, we have $N\ge 0$ such that
  $f^*= \bigvee_{ k= 0}^N f^k$, and then
  \begin{align}
    \prod_{ n\ge 0} f^*\compop g_n
    &= \prod_{ n\ge 0}\Big( \bigvee_{ k= 0}^N f^k\Big)\compop g_n
      \notag \\
    &= \prod_{ n\ge 0}\Big( \bigvee_{ k= 0}^N f^k\compop g_n\Big)
      \notag \\
    &= \adjustlimits \bigvee_{ 0\le k_0, k_1,\dotsc\le N\;} \prod_{
      n\ge 0} f^{ k_n}\compop g_n \label{eq:ax4.use-ax3} \\
    &\le \adjustlimits \bigvee_{ k_0, k_1,\dotsc\ge 0\;} \prod_{ n\ge
      0} f^{ k_n}\compop g_n\,, \notag
  \end{align}
  where~\eqref{eq:ax4.use-ax3} holds because of~\eqref{eq:ax3}. 
\end{proof}

\begin{lem}%
  \label{le:fomega}
  For $f\in \E$, $f^\omega\in \U$ is given by
  \begin{equation*}
    f^\omega( x, t)=
    \begin{cases}
      \ltrue &\text{if } x\ne \bot, t= \infty, \text{and } \exists
      t_0\in[ 0, \infty]: f( x, t_0)\ge x\,; \\
      \ltrue &\text{if } x\ne \bot, t\ne \infty, \text{and } \exists
      t_0\le t: f( f( x, t_0), 0)\ge f( x, t_0)\ne \bot\,; \\
      \lfalse &\text{otherwise}\,.
    \end{cases}
  \end{equation*}
\end{lem}

\begin{proof}
  The situation is clear for $f= \bot$ or $x= \bot$, so we can assume
  $f\ne \bot$ and $x\ne \bot$.  Let $A$ be an RTEA which computes $f$.

  Assume first that $t\ne \infty$.  In that case,
  $f^\omega( x, t)= \ltrue$ iff there is an infinite sequence
  $t_0, t_1,\dotsc\in \Realnn$ whose partial sums converge to $t$:
  $\sum_{ n= 0}^\infty t_n= t$, and such that for all $n\ge 0$,
  $f( t_n)\circ\dotsm\circ f( t_0)( x)\ne \bot$.  By convergence, we
  have $\lim_{ n\to \infty} t_n= 0$.

  By piecewise linearity, we can write
  $L\times \Realni= \bigcup_{ i= 1}^N X_i$ and
  $f= \bigvee_{ i= 1}^N f_i$, such that $f_i( y, u)= a_i u+ b_i y+
  p_i$ for $( y, u)\in X_i$ and $f_i( y, u)= \bot$ for $( y, u)\notin
  X_i$.  By construction, $p_i\le 0$ for all $i$.

  Let $\alpha= \max\{ p_i\mid i= 1,\dotsc, N\}$, then $\alpha\le 0$,
  and $\alpha< 0$ iff the prices along all paths through $A$ are
  non-zero.  Let $i\in\{ 1,\dotsc, N\}$ be such that $( y, 0)\in X_i$
  for some $y$, then $b_i= 1$ (if no time is available, we cannot
  delay in any states).  By right-continuity,
  $\lim_{ u\to 0} f_i( y, u)= f_i( y, 0)= y+ p_i$; hence if
  $\alpha< 0$, then there is $n\ge 0$ such that
  $f( t_n)\circ\dotsm\circ f( t_0)( x)= \bot$.  If $\alpha= 0$ on the
  other hand, then we can choose $t_0= t$ and $t_n= 0$ for $n\ge 1$.

  Now we show the claim for $t= \infty$.  If there is $t_0\in \Realni$
  for which $f( x, t_0)\ge x$, then we can assume $t_0> 0$ and put
  $t_n= t_0$ for all $n$ to show that $f^\omega( x, t)= \ltrue$.

  We now show that if $f( x, t_0)< x$ for all $t_0\in \Realni$, then
  $f^\omega( x, t)= \lfalse$.  Let
  $\alpha= \sup\{ f( x, t_0)- x\mid t_0\in \Realni\}$, then
  $\alpha< 0$ as $\Realni$ is compact.  We have
  $f( x, t_0)\le x+ \alpha$ for all $t_0\in \Realni$.  Now entering
  the RTEA $A$ for $f$ with initial energy lower than $x$ can disable
  some paths, but will not enable any new behavior, hence for
  $x'\le x$ and any $t_1\in \Realni$,
  $f( x', t_1)\le f( x, t_1)+ x'- x$.  Hence
  $f( t_1)\circ f( t_0)( x)\le f( x, t_1)+ f( x, t_0)- x\le f( x,
  t_1)+ \alpha\le x+ 2 \alpha$
  for all $t_0, t_1\in \Realni$.  By induction, we see that for all
  infinite sequences $t_0, t_1,\dotsc\in \Realni$ and all $n\ge 0$,
  $f( t_n)\circ\dotsm\circ f( t_0)( x)\le x+ n \alpha$.  By
  $\alpha< 0$, $f^\omega( x, t)= \lfalse$.  We have shown that
  $f^\omega( x, t)= \ltrue$ iff there exists $t_0\in \Realni$ with
  $f( x, t_0)\ge x$.
\end{proof}

\section{Decidability}

We can now apply the results of Section~\ref{se:weightedaut} to see
that our decision problems as stated at the end of
Section~\ref{se:rtea} are decidable.  Let $A=( S, s_0, F, T, r)$ be an
RTEA, with matrix representation $( \alpha, M, K)$, and
$x_0, t, y\in \Realni$.

\begin{thm}
  There exists a finite run
  $( s_0, x_0, t)\leadsto\dotsm\leadsto( s, x, t')$ in $A$ with
  $s\in F$ iff $| A|( x_0, t)> \bot$.
\end{thm}

\begin{thm}
  There exists a finite run
  $( s_0, x_0, t)\leadsto\dotsm\leadsto( s, x, t')$ in $A$ with
  $s\in F$ and $x\ge y$ iff $| A|( x_0, t)\ge y$.
\end{thm}

\begin{thm}
  There exists $s\in F$ and an infinite run
  $( s_0, x_0, t)\leadsto( s_1, x_1, t_1)\leadsto\dotsm$ in $A$ in
  which $s_n= s$ for infinitely many $n\ge 0$ iff
  $\| A\|( x_0, t)= \top$.
\end{thm}

\begin{thm}%
  \label{th:decidable}
  Problems~\ref{prb:reach},~\ref{prb:cover} and~\ref{prb:buchi} from
  Section~\ref{se:rtea} are decidable.
\end{thm}

\begin{proof}
  Let $A$ be a RTEA, then $| A|\in \E$ and $\| A\|\in \U$.  Functions
  in $\E$ are PWL, hence they can be represented using the (finitely
  many) corner points of their regions of definition together with
  their values at these corner points.

  It is clear that computable atomic RTEFs are computable piecewise
  linear (\ie~all numbers in their finite representation are
  computable), and that compositions and suprema of computable
  piecewise linear functions are again computable piecewise linear.
  Using Lemma~\ref{le:locallyclosed}, we see that all operations to
  compute $| A|$ are computable.  This shows that the theorem for
  problems~\ref{prb:reach} and~\ref{prb:cover}.

  To show the claim for $\| A\|$, we note that because of piecewise
  linearity, the criteria in Lemma~\ref{le:fomega} are decidable;
  hence if $f\in \E$ is computable, then so is $F^\omega$.  Only
  $\vee$, $\compop$ and $^\omega$ operations are used to compute $\|
  A\|$, hence also problem~\ref{prb:buchi} is decidable.
\end{proof}

\section{Conclusion}

We have developed an algebraic methodology for deciding reachability
and B{\"u}chi problems on a class of weighted real-time models where
the weights represent energy or similar quantities.  The semantics of
such systems is modeled by real-time energy functions which map
initial energy of the system and available time to the maximal final
energy level.  We have shown that these real-time energy functions
form a $^*$-continuous Kleene $\omega$-algebra, which entails that
reachability and B{\"u}chi acceptance can be decided in a static way
which only involves manipulations of energy functions.

We have seen that the necessary manipulations of real-time energy
functions are computable, and in fact we conjecture that our method
leads to an exponential-time algorithm for deciding reachability and
B{\"u}chi acceptance in real-time energy automata.  This is due to the
fact that operations on real-time energy functions can be done in time
linear in the size of their representation, and the representation
size of compositions and suprema of real-time energy functions is a
linear function of the representation size of the operands.  In future
work, we plan to do a careful complexity analysis which could confirm
this result and to implement our algorithms to see how it fares in
practice.

This paper constitutes the first application of methods from Kleene
algebra to a timed-automata like formalism.  In future work, we plan
to lift some of the restrictions of the current model and extend it to
allow for time constraints and resets {\`a} la timed automata.  We
also plan to extend this work with action labels, which algebraically
means passing from the semiring of real-time energy functions to the
one of formal power series over these functions.  In applications,
this means that instead of asking for existence of accepting runs, one
is asking for controllability.

\bibliographystyle{alpha}
\bibliography{mybib}

\end{document}